\documentclass[12pt, draftclsnofoot, journal, letterpaper, onecolumn]{IEEEtran}
\makeatletter
\def\ps@headings{%
\def\@oddhead{\mbox{}\scriptsize\rightmark \hfil \thepage}%
\def\@evenhead{\scriptsize\thepage \hfil \leftmark\mbox{}}%
\def\@oddfoot{}%
\def\@evenfoot{}}
\makeatother 

\IEEEoverridecommandlockouts
\usepackage{bbm}
\usepackage{amsfonts}
\usepackage[dvips]{graphicx}
\usepackage{times}
\usepackage{cite}
\usepackage{amsmath}
\usepackage{array}
\usepackage{amssymb}

\usepackage{stfloats}
\usepackage{slashbox}
\usepackage{graphicx}
\usepackage{footnote}
\usepackage{amsthm}
\usepackage{booktabs}
\usepackage{array}
\usepackage{algorithmic}
\usepackage{algorithm}
\usepackage{subeqnarray}
\usepackage{cases}
\usepackage{threeparttable}
\usepackage{color}
\usepackage{hyperref}
\usepackage{epstopdf}
\usepackage{cleveref}
\usepackage{bm}
\usepackage{subcaption}
\usepackage{dsfont}

\usepackage[labelformat=simple]{subcaption}

\newtheorem{proposition}{Proposition}

\begin{document}

\title{Content-Centric Sparse Multicast Beamforming for Cache-Enabled Cloud RAN}

\author{Meixia~Tao, Erkai~Chen, Hao~Zhou and Wei~Yu\\
\thanks{This paper was presented in parts at IEEE/CIC ICCC 2015 \cite{usercentric_ICCC15} and IEEE GLOBECOM 2015 \cite{contentcentric_Globecom15}.
M.~Tao and E.~Chen are with the Department of Electronic Engineering at Shanghai Jiao Tong University, Shanghai, P. R. China (emails: \{mxtao, cek1006\}@sjtu.edu.cn).
H.~Zhou was with the Department of Electronic Engineering at Shanghai Jiao Tong University, Shanghai, P. R. China and is now with Department of Electrical Engineering and Computer Science, Northwestern University, USA (email: haozhou2015@u.northwestern.edu).
W.~Yu is with the Department of Electrical and Computer Engineering, University of Toronto, Canada (email: weiyu@comm.utoronto.ca).}
}

\maketitle

\begin{abstract}
This paper presents a content-centric transmission design in a cloud radio access network (cloud RAN) by incorporating multicasting and caching. Users requesting a same content form a multicast group and are served by a same cluster of base stations (BSs) cooperatively. Each BS has a local cache and it acquires the requested contents either from its local cache or from the central processor (CP) via backhaul links. We investigate the dynamic content-centric BS clustering and multicast beamforming with respect to both channel condition and caching status. We first formulate a mixed-integer nonlinear programming problem of minimizing the weighted sum of backhaul cost and transmit power under the quality-of-service constraint for each multicast group. Theoretical analysis reveals that all the BSs caching a requested content can be included in the BS cluster of this content, regardless of the channel conditions. Then we reformulate an equivalent sparse multicast beamforming (SBF) problem. By adopting smoothed $\ell_0$-norm approximation and other techniques, the SBF problem is transformed into the difference of convex (DC) programs and effectively solved using the convex-concave procedure algorithms. Simulation results demonstrate significant advantage of the proposed content-centric transmission. The effects of heuristic caching strategies are also evaluated.
\end{abstract}


\section{Introduction}

Owing to the popularity of smart mobile devices, the mobile data traffic is growing rapidly. Meanwhile, the provision of the types of wireless services is also experiencing a fundamental shift from the traditional connection-centric communications, such as phone calls, e-mails, and web browsing, to the emerging content-centric communications, such as video streaming, push media, mobile applications download/updates, and mobile TV \cite{Cisco2015}. A central feature of these emerging services is that a same copy of content may be needed by multiple mobile users, referred to as \emph{content diversity} \cite{LiuHui-WCM14} or \emph{content reuse} \cite{Femtocaching-mag13}. Two enabling techniques to exploit such content diversity are multicasting and caching. Compared with point-to-point unicast transmission, point-to-multipoint multicast transmission provides an efficient capacity-offloading approach for common content delivery to multiple subscribers on a same resource block\cite{eMBMS, multicast06}. Caching, on the other hand, brings contents closer to users by pre-fetching the contents during the off-peak time and hence can greatly reduce the network congestion and improve the user-perceived experience \cite{Femtocaching-mag13, content_caching, proactive_caching}. In this paper, we propose a content-centric transmission design based on cloud radio access network (cloud RAN) architectures for efficient content delivery by integrating both multicasting and caching.

Cloud RAN is a promising network architecture to boost network capacity and energy efficiency \cite{CRAN_mag}. In a cloud RAN, the base stations (BSs) are all connected to a central processor (CP) via digital backhaul links, thus enabling joint data processing and precoding across multiple BSs \cite{cooperative_jsac}. However, performing full joint processing requires not only signalling overhead but also payload data sharing among all the BSs, resulting in tremendous burden on backhaul links.
To alleviate the backhaul capacity requirement in cloud RAN architecture, one way is to associate each user with a cluster of BSs so that each user is cooperatively served by the given cluster of BSs through joint precoding. With BS clustering, each user's data only needs to be distributed to its serving BSs from CP rather than all the BSs, thus the overall backhaul load can be greatly reduced. Dynamic BS clustering and the associated sparse beamforming have been developed in \cite{ZQLuo_jsac, Tony_trans, WeiYu_globecom, WeiYu_access, smooth_trans, JunZhang_trans}.
These previous works present a user-centric view on the BS clustering and beamforming, regardless of the content diversity. In practice, however, the mobile users normally send requests in a non-uniform manner, following certain content popularity distribution, e.g., the Zipf distribution \cite{zipf_distribution}. Popular contents are likely to be requested by multiple users.

To exploit such content popularity, we propose a content-centric BS clustering and multicast beamforming in this paper. We first group all the scheduled users according to their requested contents. In specific, the users who request a same content form one group. Then we use multicast transmission to deliver each requested content to the corresponding user group. The users from the same multicast group receive a common content sent by a cluster of BSs. The BS clustering is designed with respect to each requested content. The BS clusters for different contents may overlap. Compared with traditional unicast transmission in user-centric BS clustering, multicast transmission in the considered content-centric BS clustering can improve energy and spectral efficiency, thus providing efficient content delivery in wireless networks.

In addition to multicasting, we incorporate caching in the considered cloud RAN architecture to facilitate content-centric communications. Wireless caching has been recently proposed as a promising way of reducing peak traffic and backhaul load, especially for video content delivery. Due to the content reuse feature of video streaming, i.e., many users are likely to request the same video content, caching some of the popular contents at the local BSs during the off-peak time  \cite{Femtocaching_trans} or pushing them at user devices directly through broadcasting \cite{wang-chen-liu14, CoMPcloud_SCVT15} can help improve the network throughput and user-perceived quality of experience.
In the considered cache-enabled cloud RAN, each BS is equipped with a local cache. If the content requested by a user group is already cached in its serving BS, the serving BS will transmit the content directly. Otherwise, the serving BS needs to fetch the content via the backhaul links from the CP. Compared with non-cache cooperative networks, cache-enabled cloud RAN can fundamentally reduce the backhaul cost and support more flexible BS clustering. A similar wireless caching network has been considered in \cite{JunZhang_pimrc}, but it is based on a user-centric transmission design.

In this paper, we investigate the joint design of content-centric BS clustering and multicast beamforming in the considered cache-enabled cloud RAN to improve the network performance as well as to reduce backhaul cost. There are two important issues to address in this paper. The first one is \textbf{how to optimize the content-centric BS clustering}.
In the traditional user-centric BS clustering without cache \cite{WeiYu_access, JunZhang_trans}, each user is most likely to be served by a cluster of BSs which are nearby and have good channel conditions. While in the proposed content-centric BS clustering, there are multiple users receiving a same content. Since the users are geographically separated from each other, the BS clustering becomes more involved. Moreover, after introducing cache at each BS, the BS clustering needs to be adaptive to the caching state as well. Therefore, the BS clustering in the considered network must be both channel-aware and cache-aware.

To address this issue, we formulate an optimization problem with the objective of minimizing the total network cost subject to the quality-of-service (QoS) constraint for each muticast group. The total network cost is modeled by the weighted sum of backhaul cost and transmit power cost, where the backhaul cost is counted by the accumulated data rate transferred from the CP to all the BSs through the backhaul links. This problem is a mixed-integer nonlinear programming (MINLP) problem. Through theoretical analysis, we show that all the BSs which cache the content requested by a multicast group can be included in the BS cluster of this content without loss of optimality, regardless of their channel conditions. This finding can be used to reduce the search space for the global optimal solution of the joint content-centric BS clustering and multicast beamforming problem.

To make the problem more tractable, we reformulate an equivalent sparse multicast beamforming design problem. Solving the equivalent problem is, however, still challenging due to both the nonconvex QoS constraints and the nonconvex discontinuous $\ell_0$-norm in the objective. In sparse signal processing, one approach to handle the $\ell_0$-norm minimization problem is to approximate the $\ell_0$-norm with its weighted $\ell_1$-norm \cite{Boyd_reweighted_l1} and update the weighting factors iteratively. Another approach is to approximate the $\ell_0$-norm with smooth functions \cite{smooth_trans}, where the authors use Gaussian family functions. The smooth function method is a better approximation to the $\ell_0$-norm but its performance highly depends on the smoothness factor. In this paper, we adopt the smoothed $\ell_0$-norm approximation and compare the performance of three smooth functions: logarithmic function, exponential function, and arctangent function. The approximated problem is then transformed into the difference of convex (DC) programming problems. Two specific forms of DC programming are obtained. One has DC objective and convex constraints by using semi-definite relaxation (SDR) approach, and the other is a general one with DC forms in both objective and constraints. The convex-concave procedure (CCP) \cite{yuille2003cccp, smola2005kernel, lipp2014variations} based algorithms are then proposed to find effective solutions of the original problem.
Note that the proposed CCP-based algorithms can also be used to solve the traditional multi-group multicast beamforming problems formulated in \cite{Luo_MultlGroup_multicast}.

The second issue is \textbf{how would different caching strategies affect the overall performance of the cache-enabled cloud RAN}. In order to increase cache hit rate, i.e., the probability that a requested content can be accessed at the local cache of its delivering BSs, caching strategies should be designed carefully. It is noted in \cite{molisch2014caching} that when BSs are sparsely deployed such that each user can only be connected to one BS, each BS should cache the most popular contents, otherwise when BSs are densely deployed such that each user can be served by multiple BSs, the optimal caching strategy is highly complex.
Notice that the cache placement and content delivery phases happen on different timescales: cache placement in general is in a much larger timescale (e.g., on a daily or hourly basis) while content delivery is in a much shorter timescale \cite{Femtocaching_trans, molisch2014caching}.
In \cite{Mixed-timescale}, the authors studied the mixed-timescale precoding and cache control in MIMO interference network. But the caching strategy in each transmitter is the same. In addition, the precoding is limited to two modes only, i.e., each user is either served by all the BSs simultaneously or served by one of the BSs only.

In this work, although the content placement is assumed to be given in the formulated content-centric sparse multicast beamforming problem, we shall briefly address the caching strategy problem through simulation. We consider three heuristic caching strategies, popularity-aware caching, probabilistic caching, and random caching. Numerical results have demonstrated interesting insights into the performance of these caching strategies together with the proposed sparse beamforming algorithms.

The rest of the paper is organized as follows. Section \ref{section:Network Model} introduces the network model and assumptions. Section \ref{section:Problem Formulation and Analysis} provides the formulations of content-centric BS clustering and multicast beamforming problem and the equivalent sparse beamforming problem. The CCP based algorithms are presented in Section \ref{section:CCP based Algorithm Design}. Comprehensive simulation results are provided in Section \ref{section:Simulation Results}. Finally, we conclude the paper in Section \ref{section:Conclusion}.

\emph{Notations}: Boldface lower-case and upper-case letters denote vectors and matrices respectively. Calligraphy letters denotes sets. $\mathbb{R}$ and $\mathbb{C}$ denote the real and complex domains, respectively. $\mathbb{E}(\cdot)$ denotes the expectation of a random variable. $\mathcal{CN}(\mu,\sigma^2)$ represents a complex Gaussian distribution with mean $\mu$ and variance $\sigma^2$. The conjugate transpose and $\ell_p$-norm of a vector are denoted as $(\cdot)^H$ and $\lVert \cdot \rVert_p$ respectively. $\mathbf{1}_M$ and $\mathbf{0}_M$ denote the $M$-long all-ones and $M$-long all-zeros vectors respectively. The inner product between matrices $\mathbf{X}$ and $\mathbf{Y}$ is defined as $\left\langle \mathbf{X}, \mathbf{Y}\right\rangle = \text{Tr}(\mathbf{X}^H\mathbf{Y})$. For a square matrix $\mathbf{S}_{M \times M}$, $\mathbf{S}\succeq \mathbf{0}$ means that $\mathbf{S}$ is positive semi-definite. The real part of a complex number $x$ is denoted by $\mathfrak{R}\{x\}$.

\section{Network Model and Assumptions} \label{section:Network Model}

\subsection{System Model}
Consider the downlink transmission of a cache-enabled cloud RAN with $N$ multiple-antenna BSs and $K$ single-antenna mobile users. Each BS is connected to the CP via a limited-capacity backhaul link. The CP can access a database that contains a total number of $F$ contents with equal size. Let ${\cal N}=\{1, 2, \ldots, N\}$ denote the set of BSs, where each BS is equipped with $L$ transmit antennas and has a local cache with finite storage size. At the beginning of a transmission time interval, each user submits a content request according to certain demand probabilities. Users requesting the same content are grouped together and served using multicast transmission. The transmission time interval is assumed to contain enough number of transmission frames for the system to complete the content delivery. Let the total number of multicast groups be denoted as $M$ ($1\le M \le \min\{K, F\})$ and the set of users in each group $m$ be denoted as $\mathcal{G}_m$, for $m=1, \ldots, M$. We assume that each user can request at most one content at a time, thus we have $\mathcal{G}_i \bigcup \mathcal{G}_j = \emptyset, i \neq j $ and $\sum_{m=1}^M \lvert \mathcal{G}_m \rvert \le K$.

We consider the dynamic content-centric BS clustering and multicast beamforming on a transmission frame basis. The channel remains constant within each transmission frame but varies from one frame to another.
Each multicast group $m$ is served by a cluster of BSs cooperatively during each frame, denoted as $\mathcal{Q}_m$, where $\mathcal{Q}_m \subseteq \mathcal{N}$ and they may overlap with each other. Each BS in a cluster acquires the requested contents either from its local cache or from the database in the CP through the backhaul. During each transmission frame, the BS clusters $\{\mathcal{Q}_m \}_{m=1}^M$ are dynamically optimized by the CP.
An example is shown in Fig.~\ref{fig:systemmodel}, where three mutlicast groups are formed and the instantaneous BS clusters serving the three groups are $\mathcal{Q}_1 = \{ 1,2 \}$, $\mathcal{Q}_2 = \{ 2 \}$, and $\mathcal{Q}_3 = \{ 1,2,3 \}$, respectively.

Define a binary BS clustering matrix $\mathbf{S} \in \{0,1\}^{M \times N} $, where $s_{m,n} = 1$ indicates that the $n$-th BS belongs to the serving cluster for the $m$-th multicast group and 0 otherwise. That is, $s_{m,n} = 1$ if $n \in \mathcal{Q}_m$ and $s_{m,n} = 0$ if $n \notin \mathcal{Q}_m$.
Denote the aggregate network-wide beamforming vector of group $m$ from all BSs as $\mathbf{w}_{m} = [\mathbf{w}_{m,1}^H,\mathbf{w}_{m,2}^H,\cdots, \mathbf{w}_{m,N}^H ]^H \in \mathbb{C}^{NL\times 1}$, where  $\mathbf{w}_{m,n} \in \mathbb{C}^{L\times 1}$ is the beamforming vector for group $m$ from BS $n$.
Note that $\mathbf{w}_{m,n}=\mathbf{0}$ if $n \notin \mathcal{Q}_m$. Therefore, for each group $m$, the network-wide beamformer $\mathbf{w}_{m}$ can have a sparse structure. The received signal at user $k$ from group $\mathcal{G}_m$ can be written as
\begin{equation}
y_{k} = \mathbf{h}_k^H \mathbf{w}_m x_m + \sum_{j \neq m}^M \mathbf{h}_k^H \mathbf{w}_j x_j + n_k,~\forall k \in \mathcal{G}_m
\end{equation}
where $\mathbf{h}_k \in \mathbb{C}^{NL\times 1}$ is the network-wide channel vector from all the BSs to user $k$, $x_m \in \mathbb{C}$ is the data symbol of  the content requested by group $m$ with $\mathbb{E} \left[ \lvert x_m\rvert^2 \right] = 1$, and $n_k \sim \mathcal{CN}(0,\sigma_k^2)$ is the additive white Gaussian noise at user $k$.
The received SINR for user $k\in \mathcal{G}_m$ is expressed as
\begin{equation}
  \text{SINR}_{k} = \frac {\lvert \mathbf{h}_k^H \mathbf{w}_m \rvert^2} {\sum_{j \neq m}^M \lvert \mathbf{h}_k^H \mathbf{w}_j \rvert^2 + \sigma_k^2}, ~\forall k \in \mathcal{G}_m.
\end{equation}

We define the target SINR vector as $\bm{\gamma} = [\gamma_1,\gamma_2,\cdots,\gamma_M]$ with each element $\gamma_m$ being the minimum received SINR required by the users in group $m$. In this paper, we consider the fixed rate transmission as in \cite{WeiYu_globecom}, where the transmission rate for group $m$ is set as $R_m= B \log_2(1+\gamma_m)$, where $B$ is the total available channel bandwidth.  Thus, to successfully decode the message, for any user $k \in \mathcal{G}_m$, its target SINR should satisfy $\text{SINR}_k \geq \gamma_m$.


\subsection{Cache Model}
Let $\mathcal{F}=\{1,2,\cdots,F\}$ represent the database of $F$ contents, each with normalized size of $1$. The local storage size of BS $n$ is denoted as $Y_n$ ($Y_n < F$), which represents the maximum number of contents it can cache.
We define a binary cache placement matrix $\mathbf{C} \in \{0,1\}^{F \times N}$, where $c_{f,n} = 1$ indicates that the $f$-th content is cached in the $n$-th BS and 0 otherwise. Due to limited cache size, $\sum_{f=1}^F c_{f,n} \le Y_n$.
As noted before, cache placement happens in a much larger timescale than scheduling and transmission. Hence we assume that the cache placement matrix $\mathbf{C}$ is given and fixed according to certain caching strategy, and focus on the optimization of content-centric dynamic BS clustering and multicast beamforming at the given caching design. Similar assumptions have been made in the previous literature, e.g. \cite{JunZhang_pimrc}.

\begin{figure}[t]
\begin{centering}
\includegraphics[scale=.35]{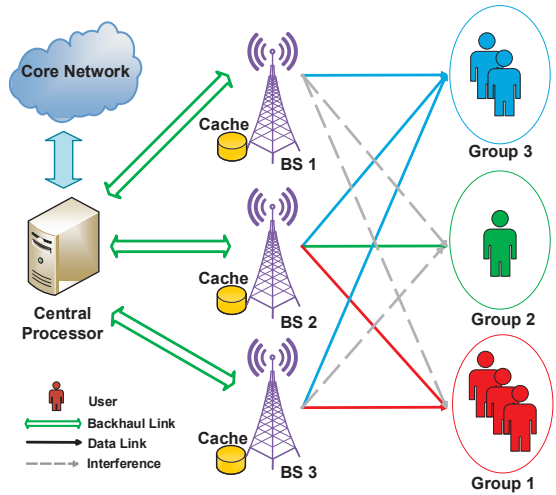}
 \caption{\small{An example of cache-enabled cloud RAN downlink.}}\label{fig:systemmodel}
\end{centering}
\end{figure}

\subsection{Cost Model}
 The total network cost for the considered network architecture consists of both the backhaul cost and the transmission power.  Let $f_m$ denote the content requested by users in group $m$.
For each BS $n \in \mathcal{Q}_m$, if content $f_m$ has been cached in its local storage, it can access the content directly without costing backhaul. Otherwise, it needs to fetch this content from the CP via the backhaul link. Since the data rate of fetching a content from the CP needs to be as large as the content-delivery rate $R_m$, we model the backhaul cost as the transmission rate of the corresponding multicast group. Thus, the total backhaul cost of the network can be expressed as:
\begin{equation}
C_B = \sum_{m=1}^M \sum_{n=1}^N s_{m,n} (1-{c}_{f_m,n}) R_m. \label{eqn:backhaul cost}
\end{equation}

The total transmission power cost at all the BSs is defined as:
\begin{equation}
C_P = \sum_{m=1}^M \sum_{n=1}^N \lVert \mathbf{w}_{m,n} \rVert_2^2. \label{eqn:power cost}
\end{equation}
As a result, the total network cost can be modeled as:
\begin{equation}
C_N = C_B + \eta C_P, \label{eqn:network cost}
\end{equation}
where $\eta > 0$ is a weighting parameter between backhaul cost and transmission power. In practice, $\eta$ can be regarded as the pricing factor to trade power for backhaul capacity.

\section{Problem Formulation and Analysis} \label{section:Problem Formulation and Analysis}

In this section, we formulate an optimization problem of minimizing the total network cost by jointly designing the content-centric BS clustering and multicast beamforming. In the considered network architecture, all the channel state information (CSI) and user requests are assumed to be available at the CP for joint processing. The cache placement is also given as assumed in the previous section.

\subsection{Joint content-centric BS clustering and multicast beamforming} \label{sec:P0}
The goal is to optimize the BS clustering and multicast beamforming at each transmission frame so as to  minimize the total network cost. This is formulated as:
\begin{subequations}
\begin{align}
\mathcal{P}_0:~\mathop{\text{minimize}}_{\{\mathbf{w}_{m,n}\}, \{s_{m,n}\}} \quad &\sum_{m=1}^M \sum_{n=1}^N s_{m,n} (1-{c}_{f_m,n}) R_m + \eta \sum_{m=1}^M \sum_{n=1}^N \lVert \mathbf{w}_{m,n} \rVert_2^2 \label{eqn:network-cost-obj}\\
\text{subject to} \quad &\text{SINR}_{k} \geq \gamma_m,~\forall k \in \mathcal{G}_m, \forall m \label{eqn:SINR-constraint} \\
& s_{m,n} \in \{0,1\},~\forall m,n \label{eqn:binary-constraint} \\
& (1 - s_{m,n}) \mathbf{w}_{m,n} = \mathbf{0},~\forall m,n \label{eqn:relationship-constraint}
\end{align} \label{eqn:network-cost-minimization-problem}
\end{subequations}
Constraint \eqref{eqn:SINR-constraint} is the minimum required SINR constraint, and constraint \eqref{eqn:relationship-constraint} indicates that if BS $n \notin \mathcal{Q}_m$, i.e., $s_{m,n}=0$, the corresponding beamformer $\mathbf{w}_{m,n}$ must be zero. Note that per-BS or per-antenna peak power constraints may be imposed in practice as in \cite{peak-power-14}. However, such peak power constraints are convex and hence will not change the nature of the formulated problem in this work as well as the algorithm design. As such we have omitted the peak power constraints and focus on the study of backhaul-power tradeoff in the total network cost.

 Problem $\mathcal{P}_0$ is combinatorial in nature. A brute-force approach to find the global optimum BS clusters is exhaustive search. In particular, there are $2^{MN}$ possible BS clustering matrices $\{ \mathbf{S}\}$. For each given BS clustering matrix $\mathbf{S}$, the backhaul cost $C_B$ becomes a constant and problem $\mathcal{P}_0$ reduces to the following power minimization problem with partially coordinated transmission:
\begin{subequations}
\begin{align}
\mathcal{P}(\mathcal{Z}_{\mathbf{S}}):~\mathop{\text{minimize}}_{ \{ \mathbf{w}_{m,n}\}} \quad & \sum_{m=1}^M \sum_{n=1}^N \lVert \mathbf{w}_{m,n} \rVert_2^2 \\
\text{subject to} \quad  & \eqref{eqn:SINR-constraint} \nonumber \\
& \mathbf{w}_{m,n} = \mathbf{0},~\forall (m,n) \in \mathcal{Z}_{\mathbf{S}}. \label{eqn:clustering-constraint}
 \end{align} \label{eqn:beamforming-problem}
\end{subequations}
where $\mathcal{Z}_{\mathbf{S}} = \{(m,n) | s_{m,n}=0 \}$ is the set of inactive BS-content associations.

Problem $\mathcal{P}(\mathcal{Z}_{\mathbf{S}})$ is similar to the QoS multi-group multicast beamforming problems \cite{Luo_MultlGroup_multicast, ZhengZheng_mulitcell_multicast} and is a nonconvex quadratically constrained quadratic programming (QCQP) problem. Unlike traditional unicast beamforming problem where the nonconvex SINR constraints can be transformed to a second-order cone programming (SOCP) problem and hence solved efficiently, the multicast beamforming problem is NP-hard in general. The authors in \cite{Luo_MultlGroup_multicast} developed semi-definite relaxation (SDR) method with randomization to obtain a good sub-optimal solution.

Once $\mathcal{P}(\mathcal{Z}_{\mathbf{S}})$ is solved for all possible BS clustering matrices $\mathbf{S}$'s, the one with the minimum objective is then selected to be the global optimal solution. Note that problem $\mathcal{P}_0$ can be infeasible if the SINR requirements $\{\gamma_m \}$ are too stringent or the channels of users in different multicast groups are highly correlated. In general, determining the feasibility of an NP-hard problem is as difficult as solving the problem itself.  In \cite{ZhengZheng_mulitcell_multicast}, a necessary condition for the QoS-based multi-cell multicast beamforming problem to be feasible is given. In this work we only discuss $\mathcal{P}_0$  when it is feasible.

The following proposition reveals some insights on BS clustering in cached-enabled networks.

\begin{proposition}
If the content $f_m$ requested by a multicast group $m$ has been cached in BS $n$, i.e., ${c}_{f_m,n} = 1$, then without loss of optimality, one can set $s_{m,n} = 1$ in problem $\mathcal{P}_0$.   \label{Proposition-P1}
\end{proposition}

\begin{proof} See Appendix A.
\end{proof}

Proposition \ref{Proposition-P1} indicates that if a certain BS caches the requested content already, then adding this BS to the existing cluster of this content regardless of its channel conditions will not cause extra backhaul cost but can potentially reduce the total transmit power because of higher degrees of freedom for cooperative transmission. This proposition however does not mean BS $n$ must serve group $m$ with strictly positive power if $c_{f_m,n}=1$.
Depending on the actual channel realizations, it is possible that the optimized beamformer $\mathbf{w}_{m,n}=\mathbf{0}$ even when $s_{m,n}=1$ according to the problem formulation in  $\mathcal{P}_0$.
Proposition \ref{Proposition-P1} can be used to reduce the exhaustive search space for global optimal BS clustering. In the extreme case when all the requested contents are cached in every BS, then the original joint BS clustering and multicast beamforming problem $\mathcal{P}_0$ reduces to the multicast beamforming problem with full cooperation, where $s_{m,n}=1, \forall m, n$.
In the general case when there exists a requested content which is not cached anywhere (except the CP), the search space for global optimal BS clustering is in the order of $2^{N}$. When the number of BSs, $N$, is very large in the considered cloud RAN architecture, the complexity can still be prohibitively high.

The above proposition holds for user-centric design as well \cite{usercentric_ICCC15}. Based on Proposition \ref{Proposition-P1}, we have developed a cache-aware greedy BS clustering algorithm, which is an extension of the greedy algorithm in \cite{usercentric_ICCC15} from unicast transmission to multicast transmission. The details are omitted due to page limit. The algorithm starts with full BS cooperation, then successively deactivates one BS from the serving cluster of a requested content based on greedy search. The BSs to be deactivated for each content only comes from the set of BSs which do not cache the content. Similar greedy algorithms without cache for user-centric BS clustering are proposed in \cite{JunZhang_trans}. Nevertheless, the number of iterations in such greedy algorithms in the worst case still grows quadratically with $MN$ and in each iteration it needs to solve a non-convex QCQP problem.

\subsection{Sparse multicast beamforming}

In this subsection, we formulate a sparse multicast beamforming (SBF) problem which is equivalent to the original problem $\mathcal{P}_0$ but more tractable. It is clear that the BS cluster matrix $\mathbf{S}$ can be specified with the knowledge of the beamformers $\mathbf{w}_{m,n}$'s. Specifically, when $\mathbf{w}_{m,n} = \mathbf{0}$,  we have:
\begin{equation}
 s_{m,n} =
 \begin{cases}
 0, &\text{if}~c_{f_m,n}=0,  \\
 0~\text{or}~1, &\text{if}~c_{f_m,n}=1. \\
 \end{cases}
\end{equation}
Otherwise when $\mathbf{w}_{m,n} \ne \mathbf{0}$, we have $s_{m,n} = 1$ from constraint \eqref{eqn:relationship-constraint}. Thus, without loss of optimality, $s_{m,n}$ can be replaced by the $\ell_0$-norm\footnote{The $\ell_0$-norm denotes the number of nonzero elements of a vector. It reduces to the indicator function in the scalar case.} of $\lVert \mathbf{w}_{m,n} \rVert_2^2$:
\begin{equation}
 s_{m,n} =  \big \| \lVert \mathbf{w}_{m,n} \rVert_2^2 \big \|_0. \label{eqn:l0-norm}
\end{equation}
By substituting \eqref{eqn:l0-norm} into the network cost in the objective \eqref{eqn:network-cost-obj},  $\mathcal{P}_0$ can be transformed into the following equivalent problem:
\begin{subequations}
\begin{align}
\mathcal{P}_{\text{SBF}}:~\mathop{\text{minimize}}_{ \{ \mathbf{w}_{m,n} \}} \quad &\sum_{m=1}^M \sum_{n=1}^N \big \| \lVert \mathbf{w}_{m,n} \rVert_2^2 \big \|_0 (1-{c}_{f_m,n}) R_m + \eta \sum_{m=1}^M \sum_{n=1}^N \lVert \mathbf{w}_{m,n} \rVert_2^2  \label{eqn:P_SBF} \\
\text{subject to} \quad  & \eqref{eqn:SINR-constraint}. \nonumber
\end{align}
\end{subequations}

Problem $\mathcal{P}_{\text{SBF}}$ is a sparse multicast beamforming problem with sparsity from the $\ell_0$-norm in the objective. This problem takes the dynamic content-centric BS clustering into account inexplicitly. Through solving this problem, a sparse beamformer for each content can be found whose nonzero entries correspond to the active serving BSs. Solving the equivalent problem $\mathcal{P}_{\text{SBF}}$ is still challenging due to the nonconvex QoS constraints and the nonconvex discontinuous $\ell_0$-norm in the objective.

\section{CCP based Sparse Multicast Beamformer Design} \label{section:CCP based Algorithm Design}
In this section, we first review some basics of the convex-concave procedure (CCP), a powerful heuristic method to find local optimal solutions to the general form of DC programming problems. After that, we show that problem $\mathcal{P}_{\text{SBF}}$ can be converted to DC programs after replacing $\ell_0$-norm with concave smooth functions and applying other techniques. Then two CCP-based algorithms are proposed to find effective solutions of $\mathcal{P}_{\text{SBF}}$.
\subsection{Basics of DC Programming and CCP} \label{subsection:Basics DC CCP}
DC programming deals with the optimization problems of functions with each represented as a difference of two convex functions. A general form of DC programming problems is written as follows:
\begin{subequations}
\begin{align}
\mathop{\text{minimize}}_{ x \in \mathbb{R}^n } \quad & g_0(x) - h_0(x)   \nonumber \\	
\text{subject to} \quad &g_i(x) - h_i(x) \leq 0, ~i=1,2,\dots,m. \nonumber
\end{align}
\end{subequations}
where $g_i$ and $h_i$ for $i=0,1,\dots,m,$ are all convex functions. A DC program is not convex unless the functions $h_i$ are affine, and is difficult to solve in general.

The convex-concave procedure is a heuristic algorithm to find a local optimal solution of DC programs. Its main idea is to convexify the problem by replacing the concave part in the DC functions, which is $h_i, i = 0,1,\dots,m$, by their first order Taylor expansions, then solve a sequence of convex problems successively. Specifically,  it starts with an initial feasible point $x_0$, i.e., $g_i(x_0) - h_i(x_0) \le 0$, for $i=1,\ldots, m$. In each iteration $t$, it solves the following convex subproblem:
\begin{subequations}
\begin{align}
\mathop{\text{minimize}}_{ x \in \mathbb{R}^n} \quad & g_0(x) - \nabla h_0(x^{(t)})^T x  \nonumber \\
\text{subject to} \quad &g_i(x) - \left [h_i(x^{(t)}) + \nabla h_i(x^{(t)})^T (x-x^{(t)}) \right ]\leq 0, ~i=1,2,\dots,m, \nonumber
\end{align}
\end{subequations}
where $x^{(t)}$ is the optimal solution obtained from the previous iteration.

The original CCP is proposed in \cite{yuille2003cccp} dealing with unconstrained or linearly constrained problems. It is then extended in \cite{smola2005kernel} to handle the general form of DC programming with DC constraints. Some other variations and extensions have been made in \cite{lipp2014variations} recently. In particular, it is shown explicitly in \cite{lipp2014variations} that the optimal solution of each iteration $x^{(t)}$ is always a feasible point of the original DC program. The convergence proof of CCP to critical points of the original problem for the differentiable case can be found in \cite{lanckriet2009convergence} and \cite{lipp2014variations}.

We would like to point out that CCP falls in the category of majorization-minimization (MM) algorithms for a particular choice of the majorization function. On the other hand, CCP can also be derived from the DC algorithm (DCA), a primal-dual subdifferential method for solving DC programs where the objective can be a difference of proper lower semi-continuous convex functions. In this paper, for differential convex functions,  we prefer CCP formulation as it is a purely primal description of the problem. In \cite{multicast-sla-2014}, a successive linear approximation (SLA) has been proposed to solve the single-group multicast beamforming problem, which can be seen as a special case of CCP. In \cite{rank-two-cccp}, the authors propose a convex inner approximation technique to tackle the max-min fairness beamforming problem in multicast relay networks, which also belongs to the class of CCP.
To the best of our knowledge, this work is the first to apply CCP to solve multi-group multicast beamforming problems.

\subsection{Smoothed $\ell_0$-norm Approximation}
To solve the sparse multicast beamforming problem $P_{\text{SBF}}$, we approximate the discontinuous $\ell_0$-norm in the objective with a continuous smooth function, denoted as $f(x)$. Specifically, we consider three frequently used smooth concave functions: logarithmic function, exponential function, and arctangent function \cite{rinaldi2010concave}, defined as
\begin{equation} \label{eqn:smooth-funtion}
f_{\theta}(x) =
\begin{cases}
\frac {\log \left (  \frac {x} {\theta} + 1 \right )} {\log(\frac{1}{\theta} + 1)}, &\text{for log-function} \\
1-\exp(-\frac{x } {\theta} ), & \text{for exp-function} \\
\arctan \left ( { \frac {x } {\theta} } \right ), &\text{for arctan-function}
\end{cases}
\end{equation}
where $\theta > 0$ is a parameter controling the smoothness of approximation. A larger $\theta$ leads to smoother function but worse approximation and vice versa. The effectiveness of these smooth functions has been demonstrated for sparse signal recovery \cite{rinaldi2010concave} and feature selection in SVM (Support Vector Machine) \cite{le2015dc}.

With the above smoothed $\ell_0$-norm, the problem $P_{\rm SBF}$ can be approximated as:
\begin{subequations} \label{problem_appx}
\begin{align}
\mathcal{P}_{1}:~\mathop{\text{minimize}}_{ \{ \mathbf{w}_{m,n} \}} \quad &\sum_{m=1}^M \sum_{n=1}^N \alpha_{m,n} f_{\theta} \left(  \lVert \mathbf{w}_{m,n} \rVert_2^2 \right ) + \eta \sum_{m=1}^M \sum_{n=1}^N\lVert \mathbf{w}_{m,n} \rVert_2^2 \\
\text{subject to} \quad &\eqref{eqn:SINR-constraint}. \nonumber
\end{align}
\end{subequations}
where $\alpha_{m,n} \triangleq (1-c_{f_m,n}) R_m$, $\forall m, n$. Note that the smooth function $f_{\theta} \left(  \lVert \mathbf{w}_{m,n} \rVert_2^2 \right )$ is concave in $\lVert \mathbf{w}_{m,n} \rVert_2^2$, but not concave in $\mathbf{w}_{m,n} $. In the following two subsections, we introduce two different techniques to convert problem $\mathcal{P}_{1}$ into DC programming problems, and then solve it using CCP-based algorithms.

\subsection{SDR-based CCP Algorithm} \label{subsection:SDR-based CCP Algorithm}
To convert $\mathcal{P}_{1}$ into a DC program, we take the SDR approach in this subsection. Define two sets of matrices $ \{\mathbf{W}_m \in \mathbb{C}^{NL\times NL} \}_{m=1}^M $ and $\{\mathbf{H}_k \in \mathbb{C}^{NL\times NL}\}_{k=1}^K $ as
\begin{equation}
\mathbf{W}_m = \mathbf{w}_{m}\mathbf{w}_{m}^H,~\forall m \quad \text{and} \quad \mathbf{H}_k = \mathbf{h}_{k} \mathbf{h}_{k}^H,~\forall k.
\end{equation}
We further define a set of selection matrices $\{\mathbf{J}_n\}_{n=1}^N$, where each $\mathbf{J}_n \in \{0,1\}^{NL\times NL}$ is a diagonal matrix defined as
\begin{equation}
\mathbf{J}_n = \text{diag} \left ( \left[ \mathbf{0}_{(n-1)L}^H, \mathbf{1}_{L}^H, \mathbf{0}_{(N-n)L}^H \right ] \right ),~\forall n.
\end{equation}
Therefore, we can rewrite $\lVert \mathbf{w}_{m,n} \rVert_2^2$ as
\begin{equation} \label{eqn:normtr}
\lVert \mathbf{w}_{m,n} \rVert^2_2 = \text{Tr} ( \mathbf{W}_{m} \mathbf{J}_n ),~\forall m,n.
\end{equation}

By removing the rank constraint ${\rm rank}\{\mathbf{W}_m\}=1$, problem $\mathcal{P}_{1}$ can be relaxed as
\begin{subequations}
\begin{align}
\mathcal{P}_{2}:~\mathop{\text{minimize}}_{ \{ \mathbf{W}_m \}} \quad &\sum_{m=1}^M\sum_{n=1}^N\alpha_{m,n}f_{\theta}(\text{Tr} \left ( \mathbf{W}_m\mathbf{J}_n\right ))   + \eta \sum_{m=1}^M \text{Tr} \left ( \mathbf{W}_{m} \right ) \label{eqn:Psdr-obj}\\
\text{subject to} \quad &\frac{ \text{Tr} ( \mathbf{W}_{m} \mathbf{H}_{k} ) } {\sum_{j \neq m}^M \text{Tr} ( \mathbf{W}_{j} \mathbf{H}_{k} ) + \sigma_k^2}  \geq \gamma_{m},~\forall k \in \mathcal{G}_m,~\forall m \label{eqn:SDR-QoS-constraint} \\
&\mathbf{W}_{m} \succeq \mathbf{0},~\forall m \label{eqn:semidefinite-constraint}
\end{align}
\end{subequations}

Clearly the SINR constraint \eqref{eqn:SDR-QoS-constraint} becomes affine. In addition, observing \eqref{eqn:Psdr-obj} closely, we find that the objective can be rewritten as the difference of two functions $g$ and $h$, defined as
\begin{equation}
g - h = \eta \sum_{m=1}^M \text{Tr} \left( \mathbf{W}_{m} \right)
- \left[-\sum_{m=1}^M\sum_{n=1}^N\alpha_{m,n}f_{\theta}(\text{Tr} \left ( \mathbf{W}_m\mathbf{J}_n\right ))\right].
\end{equation}
Recall that the smooth function $f_{\theta}(x)$ \eqref{eqn:smooth-funtion} is chosen to be concave, thus, $h$ is a convex function of $\{\mathbf{W}_m\}$. On the other hand, $g$ is affine. Then the objective of problem $\mathcal{P}_{2}$ can be expressed as a difference of convex functions. Therefore, $\mathcal{P}_{2}$ is a DC program with DC objective and convex constraints. The CCP reviewed in Section \ref{subsection:Basics DC CCP} can be readily applied with the objective function to be convexified only. The subproblem in each iteration is an SDP problem and can be solved using a generic SDP solver. The details are omitted.

If the resulting solution $\{\mathbf{W}_m\}$ after solving problem $\mathcal{P}_{2}$ is already rank-one, the optimal network-wide beamformer $\mathbf{w}^*_{m}$ of problem $\mathcal{P}_{\text{SBF}}$ can be obtained by applying eigen-value decomposition (EVD). Otherwise, the randomization and scaling method \cite{Luo_MultlGroup_multicast, ZhengZheng_mulitcell_multicast} is used to generate a suboptimal solution. In general, the SDR method can demonstrate good performance for small number of users, where the percentage of rank-one solution is high. However, as the number of users or antennas becomes very large, the probability of rank-one solution is very small, and the randomization-based solution can be far from optimal \cite{Luo_MultlGroup_multicast}. Furthermore, by adopting the SDR method, the number of variables is roughly squared  (from $MNL$ to $M(NL)^2$), which is not computationally efficient.

\subsection{Generalized CCP algorithm} \label{subsection:Generalized CCP algorithm}

In this subsection we convert $\mathcal{P}_1$ into a general form of DC programs with DC forms in both objective and constraints without any relaxation.

The nonconvex SINR constraints \eqref{eqn:SINR-constraint} in $\mathcal{P}_1$ can be rewritten as
\begin{equation}
\gamma_m \left (\sum_{j \neq m}^M \lvert \mathbf{h}_k^H \mathbf{w}_j \rvert^2 + \sigma_k^2 \right ) - \lvert \mathbf{h}_k^H \mathbf{w}_m \rvert^2 \leq 0, ~\forall k \in \mathcal{G}_m. \label{eqn:DC-SINR-constraint}
\end{equation}
Clearly, the left hand side of \eqref{eqn:DC-SINR-constraint} is a DC function. By introducing auxiliary variables $\{ t_{m,n} \in \mathbb{R} \}_{n=1,\dots, N}^{m=1,\dots, M}$ and noticing that the smooth function $f_{\theta}(x)$ is strictly monotone increasing, problem $\mathcal{P}_{1}$ can be transformed into the following problem as
\begin{subequations}
\begin{align}
\mathcal{P}_{3}:~\mathop{\text{minimize}}_{ \{ \mathbf{w}_{m,n} \}, \{ t_{m,n} \}} \quad & \sum_{m=1}^M \sum_{n=1}^N \alpha_{m,n} f_{\theta} \left(  t_{m,n} \right ) + \eta \sum_{m=1}^M \sum_{n=1}^N t_{m,n} \label{eqn:P3-obj}\\
\text{subject to} \quad &\lVert \mathbf{w}_{m,n} \rVert_2^2 - t_{m,n} \leq 0,~\forall m,n, \label{eqn:equality-constraint}\\
& \gamma_m \left (\sum_{j \neq m}^M \lvert \mathbf{h}_k^H \mathbf{w}_j \rvert^2 + \sigma_k^2 \right ) - \lvert \mathbf{h}_k^H \mathbf{w}_m \rvert^2 \leq 0, ~\forall k \in \mathcal{G}_m. \label{eqn:DC-SINR-constraint2}
\end{align}
\end{subequations}
Here, the introduction of auxiliary variables $\{ t_{m,n} \}$ is crucial. The objective \eqref{eqn:P3-obj} now becomes the difference of two convex functions expressed as:
\begin{equation}
g-h = \eta \sum_{m=1}^M \sum_{n=1}^N t_{m,n} - \left (- \sum_{m=1}^M \sum_{n=1}^N \alpha_{m,n} f_{\theta} \left(  t_{m,n} \right ) \right ).
\end{equation}
The new constraints \eqref{eqn:equality-constraint} are convex.
It can be observed that problem $\mathcal{P}_{3}$ falls into a general form of DC programming problems since both the objective and the constraints are DC functions. Hence, the general CCP algorithm reviewed in Section \ref{subsection:Basics DC CCP} can be readily applied to obtain a local optimal solution of $\mathcal{P}_{3}$.
In specific, the subproblem in the $i$th iteration of the CCP takes the following form:
\begin{subequations}
\begin{align}
\mathop{\text{minimize}}_{ \{ \mathbf{w}_{m,n} \}, \{ t_{m,n} \}} \quad & \sum_{m=1}^M \sum_{n=1}^N  \left ( \eta  + \alpha_{m,n} \nabla f_{\theta} \left(  t_{m,n}^{(i)}  \right ) \right ) t_{m,n}    \\	
\text{subject to} \quad &\lVert \mathbf{w}_{m,n} \rVert_2^2 - t_{m,n} \leq 0,~\forall m,n \\
& \gamma_m \left (\sum_{j \neq m}^M \lvert \mathbf{h}_k^H \mathbf{w}_j \rvert^2 + \sigma_k^2 \right) - \left( 2 \mathfrak{R} \left \{ (\mathbf{w}_m^{(i)} )^H \mathbf{h}_k \mathbf{h}_k^H \mathbf{w}_m \right \}- \lvert \mathbf{h}_k^H \mathbf{w}_m^{(i)}  \rvert^2 \right) \leq 0, ~\forall k \in \mathcal{G}_m
\end{align}
\end{subequations}
which is a convex QCQP problem.

\textbf{Remark 1}: By comparing with the DC transformation from $\mathcal{P}_{1}$ to $\mathcal{P}_{2}$, the DC transformation from $\mathcal{P}_{1}$ to $\mathcal{P}_{3}$ differs in two major aspects. First, the optimal solution of $\mathcal{P}_{3}$ must satisfy $\lVert \mathbf{w}^*_{m,n} \rVert_2^2 = t_{m,n}^*$ (this can be easily proved by contradiction) and thus the transformation from $\mathcal{P}_{1}$ to $\mathcal{P}_{3}$ does not incur any loss of optimality. On the other hand, the transformation from $\mathcal{P}_{1}$ to $\mathcal{P}_{2}$ is a relaxed one due to the removal of rank-one constraints. Second, while the number of variables in $\mathcal{P}_{2}$ (i.e., $M(NL)^2$) is roughly squared of that in $\mathcal{P}_{1}$ (i.e., $MNL$), the number of variables in $\mathcal{P}_{3}$ (i.e., $MN(L+1)$) almost keeps the same as that in $\mathcal{P}_1$. In the simulation section, we shall compare the performance and complexity of the two methods in greater details.

\subsection{Discussions and Algorithm Outlines} \label{sec-discussion-algorithm}
In this subsection we first provide some discussions on the initialization of the above CCP-based sparse beamforming algorithms and the updating rule of the smoothness parameter $\theta$ in \eqref{eqn:smooth-funtion}. Then we summarize our proposed algorithms formally.

\subsubsection{Initialization}

As stated in Section \ref{subsection:Basics DC CCP}, the CCP algorithm needs a feasible starting point. Unlike the single-group multicast beamforming problems \cite{multicast06, multicast-sla-2014, rank-two-cccp}, where any starting point after simple scaling can be feasible, the starting point for our problem needs to be chosen carefully. In this paper, we propose to find a feasible starting point through solving the following power minimization problem  with full BS cooperation:
\begin{subequations}
\begin{align}
\mathcal{P}_{\text{INI}}:~\mathop{\text{minimize}}_{ \{ \mathbf{W}_m \}} \quad & \sum_{m=1}^M \text{Tr} \left ( \mathbf{W}_{m} \right )\\
\text{subject to} \quad &\frac{ \text{Tr} ( \mathbf{W}_{m} \mathbf{H}_{k} ) } {\sum_{j \neq m}^M \text{Tr} ( \mathbf{W}_{j} \mathbf{H}_{k} ) + \sigma_k^2}  \geq \gamma_{m},~\forall k \in \mathcal{G}_m,~\forall m  \\
&\mathbf{W}_{m} \succeq \mathbf{0},~\forall m
\end{align}
\end{subequations}

The optimal solution $\{\mathbf{W}_m\}$ of $\mathcal{P}_{\text{INI}}$ can be used directly as a feasible starting point for the SDR-based CCP algorithm in Section \ref{subsection:SDR-based CCP Algorithm}. For the generalized CCP algorithm in Section \ref{subsection:Generalized CCP algorithm}, if $\{\mathbf{W}_m\}$ are all rank-one, then the feasible beamformers $\{\mathbf{w}_{m} \}$ can be obtained by applying EVD on $\{\mathbf{W}_m\}$. Otherwise, randomization and scaling are needed. Note that if the SDP problem $\mathcal{P}_{\text{INI}}$ turns out to be infeasible, then the original problem $\mathcal{P}_{0}$ is infeasible and both algorithms will terminate\footnote{However, even if $\mathcal{P}_{\text{INI}}$ is feasible, it does not necessarily guarantee that $\mathcal{P}_0$ is feasible.}.

The need for an initial feasible point can be removed with a penalty CCP algorithm proposed in \cite{lipp2014variations}. But the penalty algorithm is not a descent algorithm and the convergence may not be a feasible point of the original problem. The algorithm needs to be performed many times, each with a different starting point until a feasible point is obtained. This increases the overall complexity significantly.

\subsubsection{Updating Rule of $\theta$}
The performance of the smoothed $\ell_0$-norm approximation $f_{\theta}(x)$ depends on the smoothness factor $\theta$. Intuitively, when $x$ is large, $\theta$ should be large so that the approximation algorithm can explore the entire parameter space; when $x$ is small, $\theta$ should be small so that $f_{\theta}(x)$ has behavior close to $\ell_0$-norm. In our conference paper \cite{contentcentric_Globecom15}, we proposed a novel $\theta$ updating rule that achieves the above effect automatically using a sequence of $\theta$ that depends on the specific $x$ in each iteration. More specifically, $\theta$ is set to be the one that maximizes the gradient of the approximation function. It is shown in \cite{contentcentric_Globecom15} that with such updating rule, the three smooth functions in \eqref{eqn:smooth-funtion} perform almost identically.
In this work, we propose to implement an annealing design of $\theta$. We begin with a large value of $\theta$, solve $\mathcal{P}_{2}$ using SDR-CCP algorithm (or $\mathcal{P}_{3}$ using generalized CCP algorithm), then decrease $\theta$ by a given factor $\beta$ (i.e., $\theta \gets \beta \theta$) and solve $\mathcal{P}_{2}$ (or $\mathcal{P}_{3}$) again with the initial point given by the solution from the previous iteration. This scheme is then iterated until $\theta$ is sufficiently small.

\subsubsection{Algorithm Outlines}
In summary, the two proposed algorithms are outlined in Alg.~\ref{alg:SDR-CCP} and Alg.~\ref{alg:G-CCP}, respectively. Besides the differences mentioned in Remark 1, from Alg.~\ref{alg:SDR-CCP} and Alg.~\ref{alg:G-CCP} one can note that the SDP-CCP algorithm and the G-CCP algorithm also differ in randomization and scaling. In specific, the randomization and scaling in SDR-CCP is performed in the last step of Alg.~\ref{alg:SDR-CCP} and needs to be done many times in order to get a good approximate solution if $\{\mathbf{W}_m^{*}\}$ do not have rank one; however, the randomization and scaling in G-CCP is performed at the initialization step of Alg.~\ref{alg:G-CCP} and can be finished as soon as a feasible point $\{\mathbf{w}_m^{(0)}\}$ is found.

\textbf{Remark 2}: In the extreme case when $\eta\to+\infty$, the original problem $\mathcal{P}_0$ \eqref{eqn:network-cost-minimization-problem} or $\mathcal{P}_{SBF}$ \eqref{eqn:P_SBF} reduces to the total power minimization problem of multi-group multicast beamforming in a fully cooperative network subject to QoS constraints. In this case, the SDR-CCP algorithm (Alg.~1) reduces to the traditional SDR method \cite{Luo_MultlGroup_multicast} without invoking CCP. However, the G-CCP algorithm (Alg.~2) still stands as it is except there is no need to update the smoothness factor $\theta$. Thus we can claim that our proposed G-CCP based sparse muticast beamforming algorithm is general and can be used to solve the traditional multi-group multicast beamforming problem \cite{Luo_MultlGroup_multicast, ZhengZheng_mulitcell_multicast} as special cases.

\begin{algorithm}[h]
\caption{(SDR-CCP) SDR-CCP Based Sparse Multicast Beamforming Algorithm} \label{alg:SDR-CCP}
\begin{algorithmic}[0]
\STATE \textbf{Initialization:}
\begin{enumerate}
  \item Find a feasible starting point $\{\mathbf{W}_m^{(0)}\}$ by solving $\cal{P}_{\rm INI}$.
  \item Set the smoothness factor $\theta=\theta_0$, decaying factor $0< \beta < 1$ and small constant $\epsilon$.
\end{enumerate}
\STATE \textbf{Repeat}
\begin{enumerate}
   \item Solve $\mathcal{P}_{2}$ using CCP at the starting point $\{\mathbf{W}_m^{(0)}\}$ and denote the solution as $\{\mathbf{W}_m^{*}\}$
   \item Update $\theta  \gets \beta\theta$ and $\{\mathbf{W}_m^{(0)}\} \gets \{\mathbf{W}_m^{*}\}$
\end{enumerate}
\STATE \textbf{Until} $\theta < \epsilon$.
\STATE \textbf{If} ${\rm rank}(\mathbf{W}_m^{*}) = 1$, $\forall m$, apply EVD on $\{\mathbf{W}_m^{*}\}$ to obtain the final solution $\{ \mathbf{w}_m^{*} \}$ .
\STATE \textbf{Else}, apply Gaussian randomization and scaling to obtain the approximate solution $\{ \mathbf{w}_m^{*} \}$
\end{algorithmic}
\end{algorithm}

\begin{algorithm}[h]
\caption{(G-CCP) Generalized CCP Based Sparse Multicast Beamforming Algorithm} \label{alg:G-CCP}
\begin{algorithmic}[0]
\STATE \textbf{Initialization:}
\begin{enumerate}
  \item Solve $\cal{P}_{\rm INI}$ and denote the solution as $\{\mathbf{W}_m^{(0)}\}$.
  \item If ${\rm rank}(\mathbf{W}_m^{(0)}) = 1$, $\forall m$, apply EVD to obtain a feasible point $\{ \mathbf{w}_m^{(0)} \}$.
  \item Else, apply Gaussian randomization and scaling to obtain a feasible point $\{ \mathbf{w}_m^{(0)} \}$.
  \item Set the smoothness factor $\theta=\theta_0$, decaying factor $0< \beta < 1$, and small constant $\epsilon$.
\end{enumerate}
\STATE \textbf{Repeat}
\begin{enumerate}
   \item Solve $\cal{P}_{\rm 3}$ using CCP at the starting point $\{\mathbf{w}_m^{(0)}\}$ and denote the solution as $\{\mathbf{w}_m^{*}\}$
   \item Update $\theta  \gets \beta\theta$ and $\{\mathbf{w}_m^{(0)}\} \gets \{\mathbf{w}_m^{*}\}$
\end{enumerate}
\STATE \textbf{Until} $\theta < \epsilon$.
\STATE \textbf{Output} solution $\{\mathbf{w}_m^{*}\}$
\end{algorithmic}
\end{algorithm}

\begin{figure}[t]
\begin{centering}
\includegraphics[scale=.40]{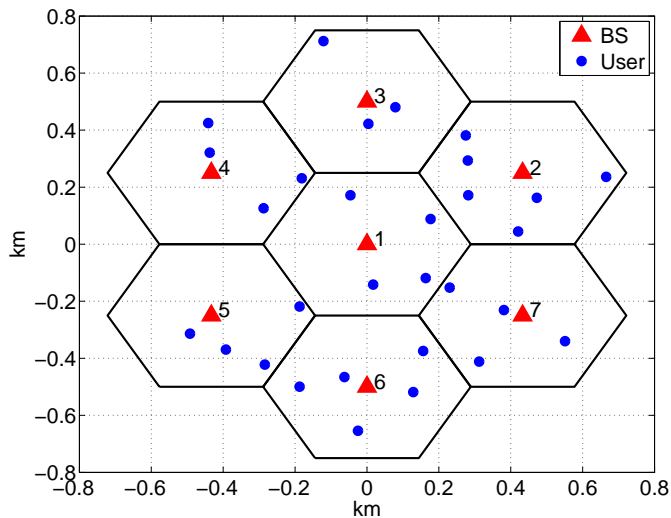}
 \caption{\small{Simulation scenario with 7 BSs, 30 mobile users randomly distributed.}}\label{fig:location}
\end{centering}
\end{figure}

\section{Simulation Results} \label{section:Simulation Results}
In this section, we provide comprehensive simulations to illustrate the performance of the proposed content-centric sparse beamforming algorithms.
We consider a hexagonal multicell cloud RAN network, where each BS is located at the center of a hexagonal-type cell as illustrated in Fig.~\ref{fig:location}. There are $N=7$ BSs in total and each BS has $L=4$ antennas. The distance between adjacent BSs is set to $500$m. The mobile users are uniformly and independently distributed in the network, excluding an inner circle of $50$m around each BS. All BSs are assumed to have equal cache size of $Y$ (e.g., $Y_n = Y,~\forall n$). The transmit antenna power gain at each BS is $10$dBi. The available channel bandwidth is $10$MHz. The distance-dependent pathloss is modeled as $PL(\text{dB}) = 148.1 + 37.6 \text{log}_{10}(d)$, where $d$ is the distance in kilometers. The log-normal shadowing parameter is assumed to be $8$dB. The small-scale fading is the normalized Rayleigh fading. The noise power spectral density $\sigma_k^2$ for all users is set to be $-172$dBm/Hz. The SINR target for each multicast group is $\gamma_m=10$dB, $\forall m$.
%

%
In the proposed updating rule of smoothness factor $\theta$, it generally requires the initial $\theta_0$ to be large enough so that the smooth function $f_\theta(x)$ can explore the entire parameter space during the iterations. In our simulation, we set $\theta_0 = \max\{\text{Tr} (\mathbf{W}^{(0)}_m\mathbf{J}_n ), \forall m, n\}$ for Alg.~\ref{alg:SDR-CCP} and $\theta_0 = \max\{\|\mathbf{w}_{m,n}^{(0)} \|_2^2, \forall m, n\}$ Alg.~\ref{alg:G-CCP}, respectively, where $\mathbf{W}^{(0)}_m$ is the optimal solution of $\mathcal{P}_{\text{INI}}$ and $\mathbf{w}_{m,n}^{(0)}$ is obtained from $\mathbf{W}^{(0)}_m$ by EVD (Gaussian randomization and scaling is applied if necessary). Our empirical results have shown that such $\theta_0$ is large enough so that further increasing it will not bring additional gain. For the decaying factor $\beta$, a larger $\beta$ can lead to more accurate result and hence better performance but the convergence speed can be slow; a smaller $\beta$ can speed up the convergence but the algorithm is more likely to get a suboptimal solution. In our simulation, we empirically choose $\beta=0.1$, which can strike a good balance between the convergence speed and the performance. The small constant $\epsilon$ is set to be $10^{-6}$.

In each simulation trial we consider $K=30$ active users and each user submits a content request independently to a database of $F=100$ contents. Each simulated result is obtained by averaging over $100$ independent simulation trials, unless stated otherwise. In each trial we only generate one set of user locations, channel realizations, and user requests for simulation simplicity. We assume the following two different content popularity distributions: 1) \emph{Unequal popularity}: among the $100$ contents, one content belongs to trending news with request probability $0.5$, and the other $99$ contents share the rest $0.5$ of the request probability following a Zipf distribution with skewness parameter $\alpha$. In general, large $\alpha$ means more user requests are concentrated on fewer popular contents. 2) \emph{Equal popularity}: all the $100$ contents are requested with equal probabilities. In our simulation the defaulted setting is the unequal popularity with $\alpha=1$. Each BS caches $Y=10$ contents if not specified otherwise.

The following three heuristic caching strategies are considered:
\begin{itemize}
  \item \emph{Popularity-aware Caching (PopC):} Each BS caches the most popular contents until its storage is full. In such caching scheme, the cached contents in all the BSs are the same if their cache sizes are the same. This strategy can bring significant opportunity for full cooperation when the content popularity distribution is highly non-uniform. On the other hand, if the popularity is equally distributed, the cache hit rate can be very low and may cause large backhaul burden.
  \item \emph{Random Caching (RanC):} Each BS caches the contents randomly with equal probabilities regardless of their popularity distribution.
  \item \emph{Probabilistic Caching (ProC):} Each BS caches a content randomly with probability depending on the content popularity, and the more popular the content is, the more likely it will be cached in each BS. This caching strategy can strike a good balance between cache hit rate and cooperative transmission gain.
\end{itemize}

\subsection{Comparison between SDR-CCP and Generalized CCP algorithms}
We first demonstrate the convergence behavior of the proposed two CCP based algorithms denoted as SDR-CCP and G-CCP, in Fig.~\ref{fig:convergence_behavior_SDR-CCP} and Fig.~\ref{fig:convergence_behavior_G-CCP}, respectively. For illustration purpose, the smoothness factor in all smooth functions is fixed at $\theta=0.01$. Popularity-aware caching is adopted. It can be seen clearly that both SDR-CCP and G-CCP converge within less than $10$ iterations for all the considered cases.

In Fig.~\ref{fig:SDR_DC_atan}, we plot the backhaul-power tradeoff curves achieved by the proposed two algorithms. The tradeoff curves are obtained by controlling the weight parameter $\eta$ between the backhaul cost and the transmit power cost. When $\eta\to 0$ (e.g., $\eta = 10^{-6}$), the total network cost only takes backhaul capacity into account. When $\eta\to+\infty$, the total network cost only counts the transmit power and the optimization problem reduces to the transmit power minimization problem $\mathcal{P}(\mathcal{Z}_{\mathbf{S}})$ with $s_{m,n} = 1,~\forall m,n$. Without loss of generality, we take the arctangent smooth function for example. It is observed that the G-CCP algorithm provides a better backhaul-power tradeoff than the SDR-CCP algorithm. In particular, at the same backhaul cost, G-CCP can achieve $1 \sim 5$ dB lower power cost than SDR-CCP. In the extreme case with transmit power minimization (i.e., $\eta\to+\infty$), our proposed problem $\mathcal{P}_0$ reduces to the multi-group QoS multicast beamforming problem in \cite{Luo_MultlGroup_multicast} and, correspondingly,  the SDR-CCP based algorithm reduces to the SDR method in \cite{Luo_MultlGroup_multicast}. It is seen that the proposed G-CCP algorithm still outperforms the SDR method in \cite{Luo_MultlGroup_multicast} by saving $0.5$dB power cost.

We also compare the simulation running time of the two CCP-based algorithms in Table~\ref{tab:complexity-eta}. The simulation is based on MATLAB R2012a and carried out on a Windows x64 machine with 3.2 GHz CPU and 4 GB RAM. We adopt the CVX package with SDPT3 solver in \cite{cvx} to solve the convex subproblem in each iteration of CCP. It is seen from Table~\ref{tab:complexity-eta} that the running time of G-CCP is around $46\% \sim 56 \% $ of the SDR-CCP algorithm in general. In the extreme case when $\eta \to+\infty$, the running time of G-CCP is only $33\%$ that of SDR-CCP. The complexity reduction of G-CCP is contributed by two facts. First, the problem size of $\mathcal{P}_3$ in G-CCP is smaller than that of ${\mathcal P}_2$ in SDR-CCP, as noted in Remark 1. Second, G-CCP requires less number of randomization and scaling steps as noted in Section~\ref{sec-discussion-algorithm}.

From the above comparison, it can be concluded that the G-CCP based algorithm is superior to the SDR-CCP algorithm in both performance and computational complexity. Therefore, in the rest of our simulation, we only use the G-CCP algorithm.

\begin{figure}[!htb]
\begin{centering}
\includegraphics[scale=.60]{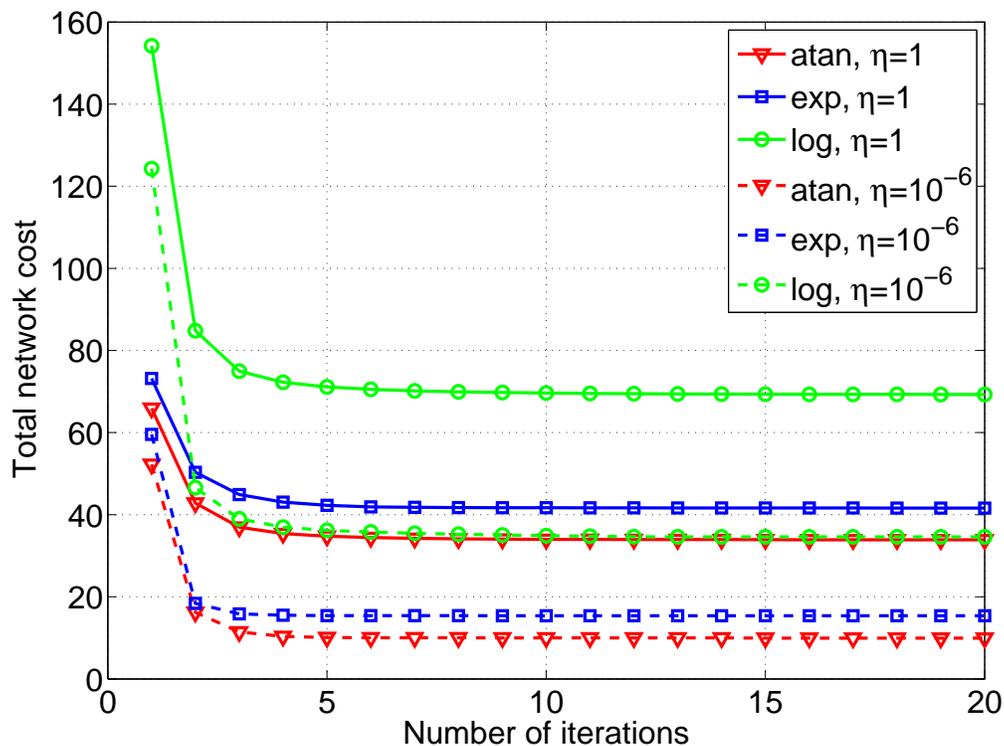}
 \caption{\small{Convergence behavior of the SDR-CCP based algorithm with $\theta = 0.01$.}}\label{fig:convergence_behavior_SDR-CCP}
\end{centering}
\end{figure}

\begin{figure}[!htb]
\begin{centering}
\includegraphics[scale=.60]{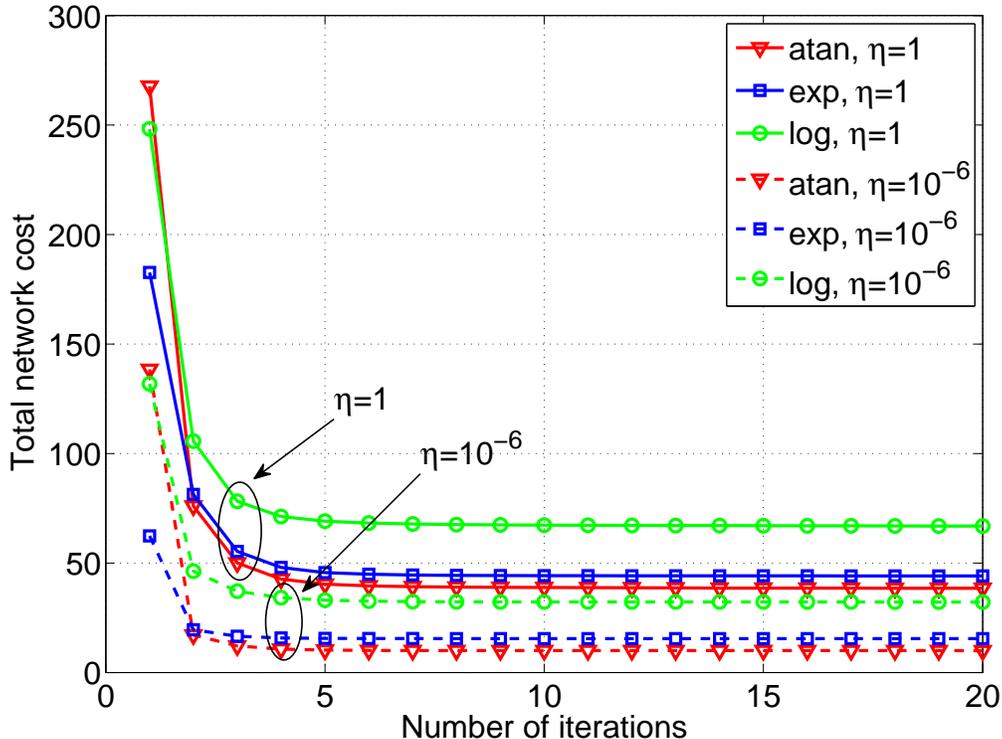}
 \caption{\small{Convergence behavior of the general CCP based algorithm with $\theta = 0.01$.}}\label{fig:convergence_behavior_G-CCP}
\end{centering}
\end{figure}

\begin{figure}[!htb]
\begin{centering}
\includegraphics[scale=.60]{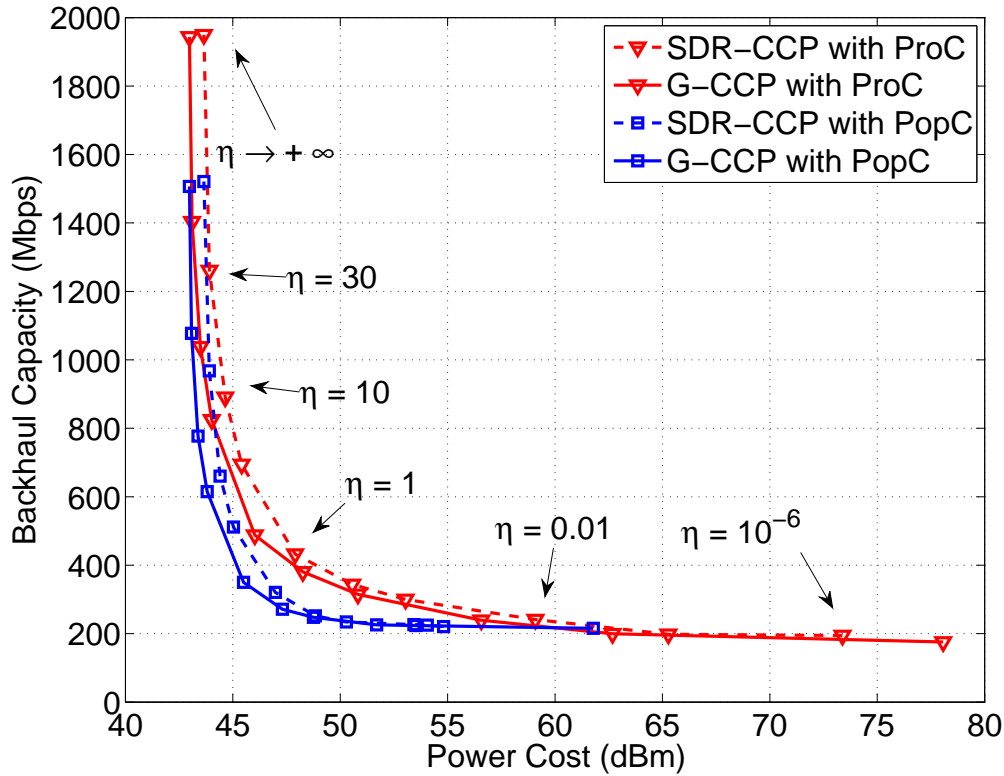}
 \caption{\small{Backhaul-power tradeoff comparison between SDR-CCP and G-CCP.}}\label{fig:SDR_DC_atan}
\end{centering}
\end{figure}

\begin{table}[!htb]
\caption{Comparison of running times (second) of SDR-CCP and G-CCP algorithms.}  \label{tab:complexity-eta}
\centering
\begin{tabular}{|c|c|c|c|c|c|c|c|c|c|c|}
\hline
$\eta$ & $10^{-6}$ & $10^{-3}$ & 0.01 & 0.1 & 0.3 & 1 & 5 & 10 & 30 & $+\infty$ \\
\hline
SDR-CCP (sec) & 909.24 & 781.02 & 705.90 & 650.14 & 642.48 & 620.15 & 582.01 & 558.53 & 482.48 & 247.97 \\
\hline
G-CCP (sec) & 510.87& 416.66 & 365.84 & 327.06 & 320.32 & 307.66 & 270.43 & 255.65 & 222.24 & 81.95 \\
\hline
\end{tabular}
\end{table}

\subsection{Comparison of Different Smooth Functions}
In Fig. \ref{fig:smooth-function}, we compare the performance of the three smooth functions with popularity-aware caching.
Each result is obtained by averaging over $200$ independent simulation trials.
It is seen that the three functions have similar performance for a wide range of $\eta$. In the extreme case when $\eta\to 0$, the arctangent function can achieve slightly lower backhaul cost than the other two functions. Thus, in the rest of our simulation, only the arctangent function is adopted.

\begin{figure}[!htb]
\begin{centering}
\includegraphics[scale=.45]{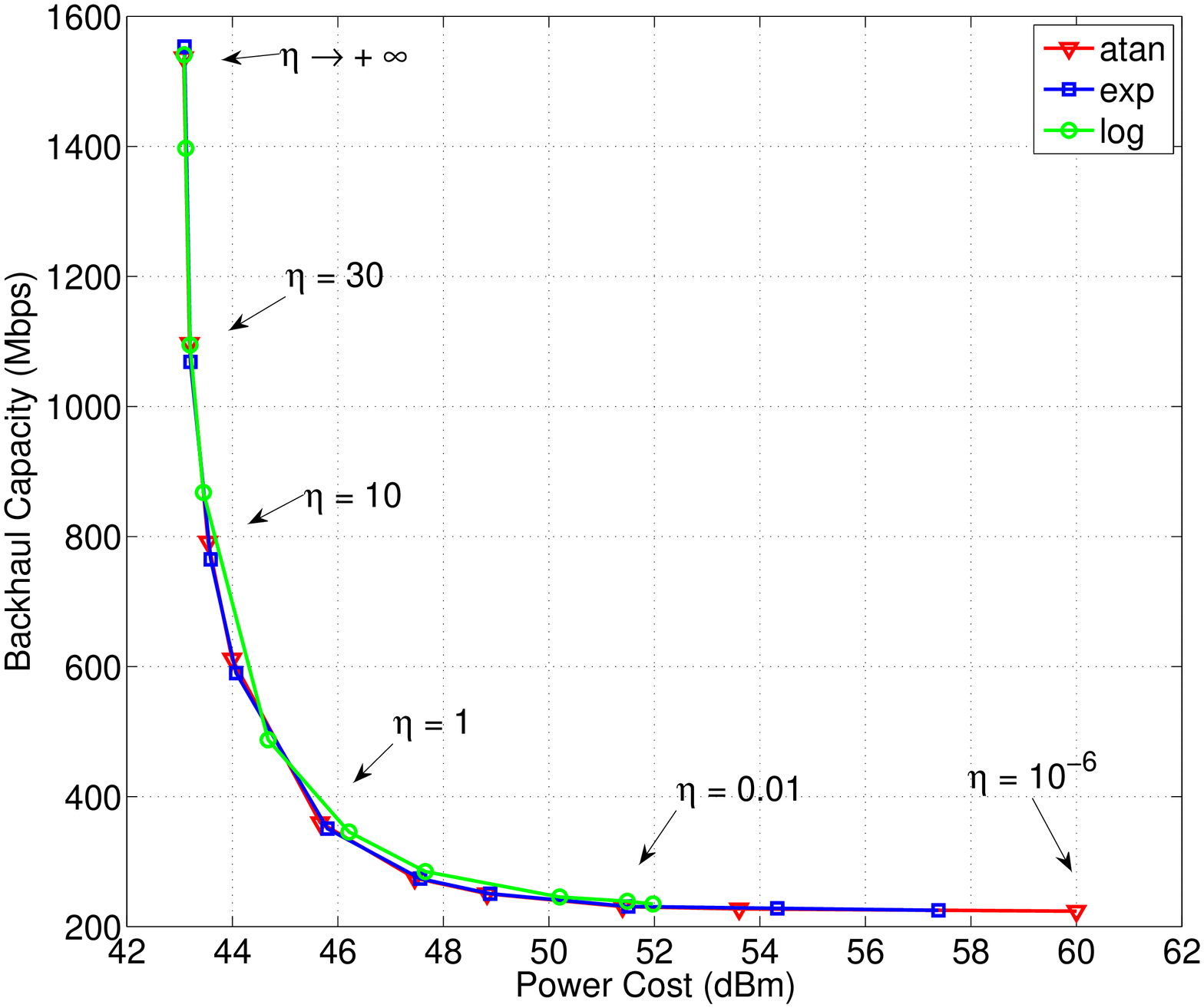}
 \caption{\small{Performance comparison of different smooth functions}}\label{fig:smooth-function}
\end{centering}
\end{figure}

\subsection{Effects of Caching}

\begin{figure}[!htb]
\begin{centering}
\includegraphics[scale=.60]{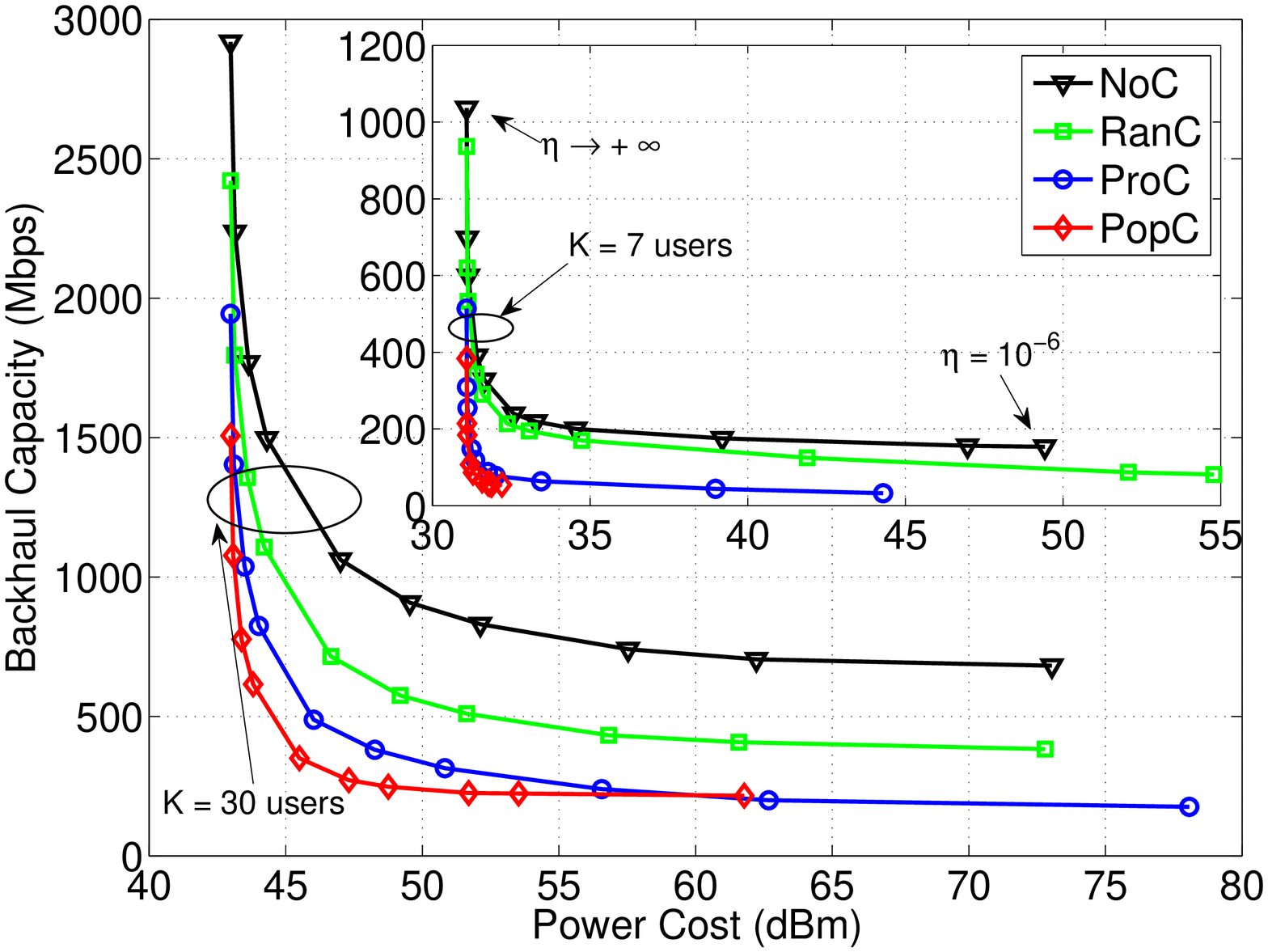}
 \caption{\small{Performance comparison of different caching strategies for unequal content popularity with $\alpha=1$.}}\label{fig:alpha1}
\end{centering}
\end{figure}

\begin{figure}[!htb]
\begin{centering}
\includegraphics[scale=.60]{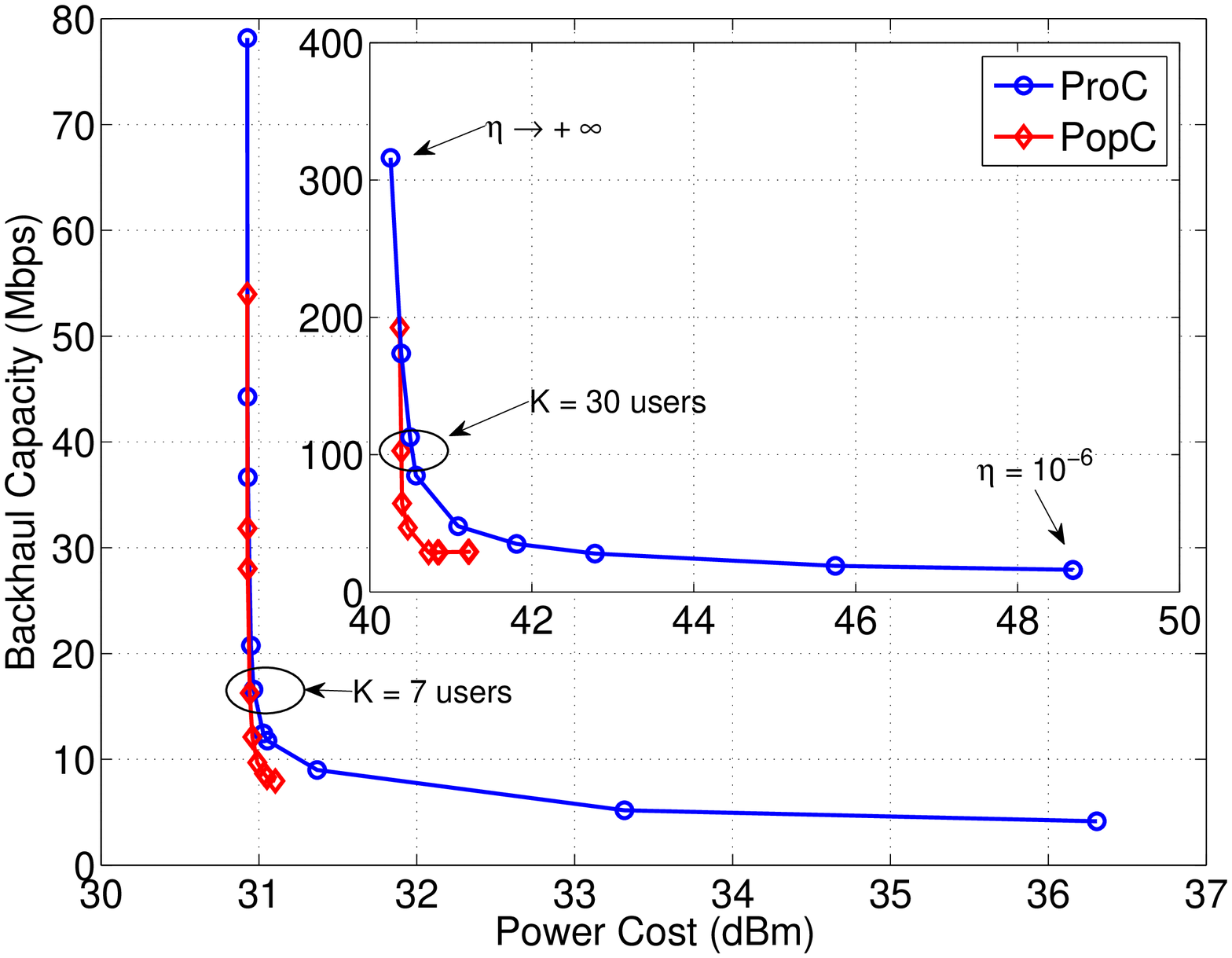}
 \caption{\small{Performance comparison of different caching strategies for unequal content popularity with $\alpha=2$.}}\label{fig:alpha2}
\end{centering}
\end{figure}

\begin{figure}[!htb]
\begin{centering}
\includegraphics[scale=.60]{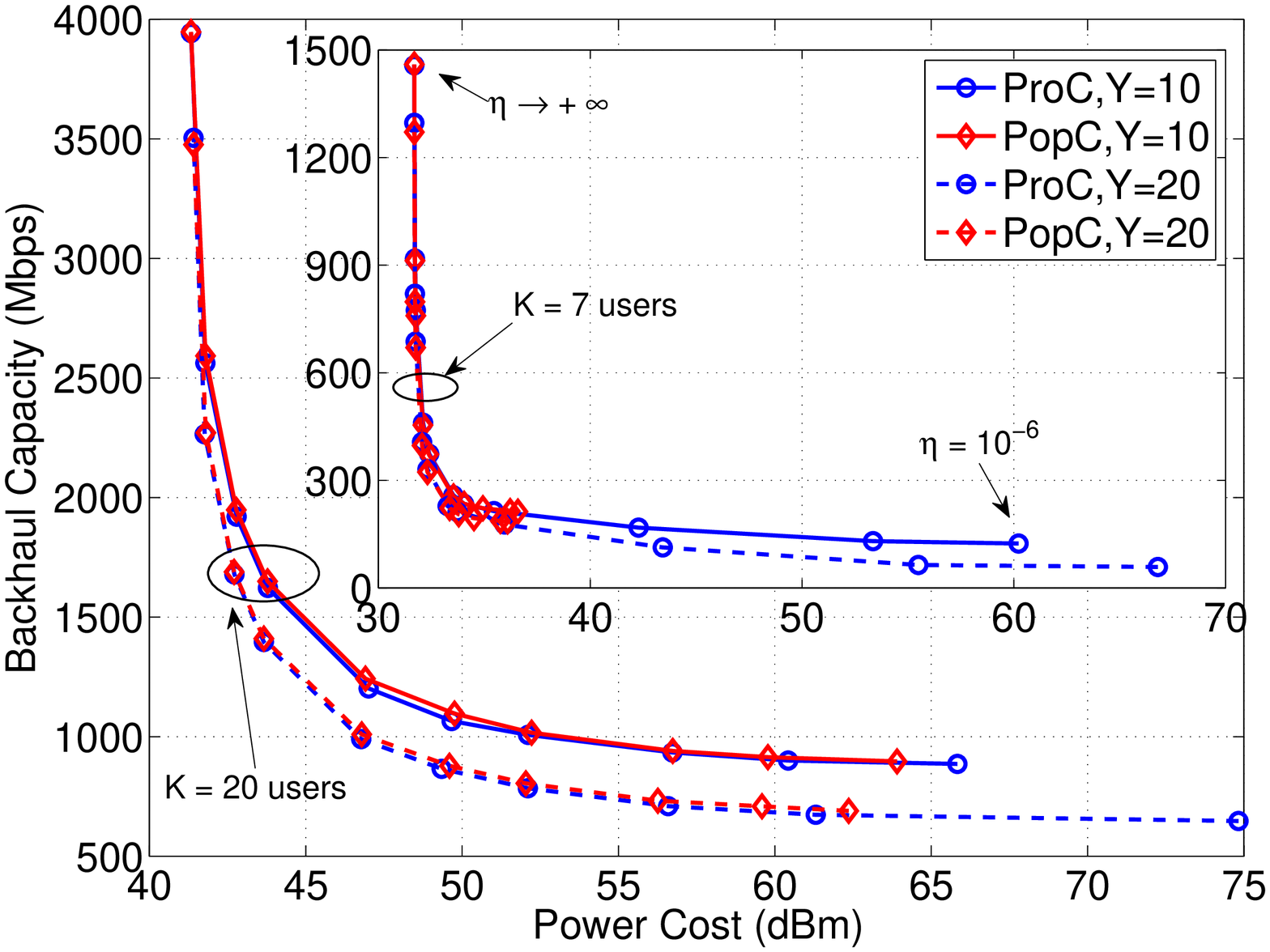}
 \caption{\small{Performance comparison of different caching strategies for equal content popularity.}}\label{fig:alpha0}
\end{centering}
\end{figure}

In Fig.~\ref{fig:alpha1}, we compare the performance of different caching strategies for unequal content popularity with skewness factor $\alpha=1$. Both $K=30$ and $K=7$ users are considered and the cache size is $Y=10$.
%
%
It is seen that compared with the network without cache, the cache-enabled network with a reasonably designed cache can dramatically reduce the backhaul cost, hence improving the backhaul-power tradeoff.

Next we focus on the comparison between the popularity-aware caching (PopC) and the probabilistic caching (ProC) under both unequal content popularity in Figs~\ref{fig:alpha1}-\ref{fig:alpha2} and equal content popularity in Fig.~\ref{fig:alpha0}. The simulation setting in Fig.~\ref{fig:alpha2} is the same as Fig.~\ref{fig:alpha1} except the skewness factor $\alpha=2$.  The results for equal content popularity in Fig.~\ref{fig:alpha0} are obtained at $K=20$ and $K=7$ users with cache size $Y=10$ and $Y=20$.

From Fig.~\ref{fig:alpha1} and Fig.~\ref{fig:alpha2} with unequal content popularity, it can be generally observed that PopC outperforms ProC for a wide range of the backhaul-power tradeoff parameter $\eta$ in the total network cost, except the limiting case $\eta\to 0$. Intuitively, since PopC places the same and most popular contents in each BS, it can enjoy the maximum transmit cooperation gain if the network is not extremely concerned with the backhaul cost. However, if the backhaul cost is the primary concern of the network (i.e., $\eta\to 0$), then ProC performs better than PopC. In specific, the achievable minimum backhaul cost of ProC is only around $50\%$-$60\%$ of that of PopC with $K=7$ users.  This is because when the user number is small, each content can be served by only a small number of BSs cooperatively to meet their SINR targets and hence it is more likely to find the requested contents in the local cache of all the serving BSs when ProC is adopted.

From Fig.~\ref{fig:alpha0} with equal content popularity, it is interestingly observed that when $K=20$, PopC and ProC (ProC is equivalent to RanC in this case) perform almost the same in the entire backhaul-power tradeoff curve, in sharp contrary to the common belief that ProC may outperform PopC because of large cache hit rate. This observation is because when the user number is large, each content should be better served by a large number of BSs cooperatively to meet their SINR targets, and hence the benefits of randomness in ProC may vanish since all the serving BSs need to access the requested content either from its local cache or from the CP via backhaul. On the other hand, when $K=7$, ProC still shows advantage over PopC by a significant reduction of achievable minimum backhaul cost as expected.

The above observations indicate that the design of more advanced caching strategies for network performance optimization should not only take into account the content popularity, but also the user density as well as the optimization objective.

Nevertheless, from \Cref{fig:alpha1,fig:alpha2,fig:alpha0}, one can see that the minimum transmit power cost of the network is the same for all caching strategies, since it is only determined by the target SINR constraint of each multicast group ($10$dB in the simulation).

\subsection{Multicast versus Unicast}

In Fig.~\ref{fig:multicast_unicast}, we compare the performance of multicast transmission and unicast transmission at different number of active users. Here, unequal content popularity with skewness factor $\alpha=1$ is assumed and the popularity-aware caching with cache size $10$ is applied. In unicast transmission, we design different beamformers for different users regardless of their requested contents.  Note that if a BS serves multiple users requesting for a same content that it does not cache, the central processor only needs to distribute to the BS one copy of the content at the maximum requested data rate. This is to ensure a fair comparison on the backhaul cost. The iterative reweighted $\ell_1$-norm based sparse unicast beamforming algorithm proposed in \cite{usercentric_ICCC15} is adopted. Such comparison between multicast transmission and unicast transmission is essentially a comparison between the proposed content-centric design and the traditional user-centric design.

It is seen from Fig.~\ref{fig:multicast_unicast} that when there are $K=30$ active users in the network, the unicast transmission (user-centric design) performs very poorly. This is mainly because there are only $N\times L=28$ transmit antennas in total in the considered cloud RAN and hence there is not enough spatial dimension to construct $30$ beamformers for unicast transmission. On the other hand, the proposed multicast transmission (content-centric design) performs well since it exploits the content popularity among different users and needs to design fewer beamformer.
When the number of active users decreases, it is observed that the backhaul-power tradeoff for unicast transmission is improved, but it is still considerably inferior to multicast transmission. In the extreme case when only power cost is concerned ($\eta \to +\infty$), it is seen that unicast transmission requires $3$ dB higher power than multicast transmission when there are $K=20$ active users. A question one may ask is whether multicast transmission still outperforms unicast transmission in the special scenario where the users requesting a same content happen to be located in geographically disjoint areas covered by different BSs. We note that such scenario belongs to the special case where the network-wide user channel vectors are orthogonal and hence unicast beamforming and multicast beamforming are equivalent.

The above observations demonstrate significant advantage of the proposed content-centric transmission design over the conventional user-centric transmission design for the considered content request model.

\begin{figure}[!htb]
\begin{centering}
\includegraphics[scale=.60]{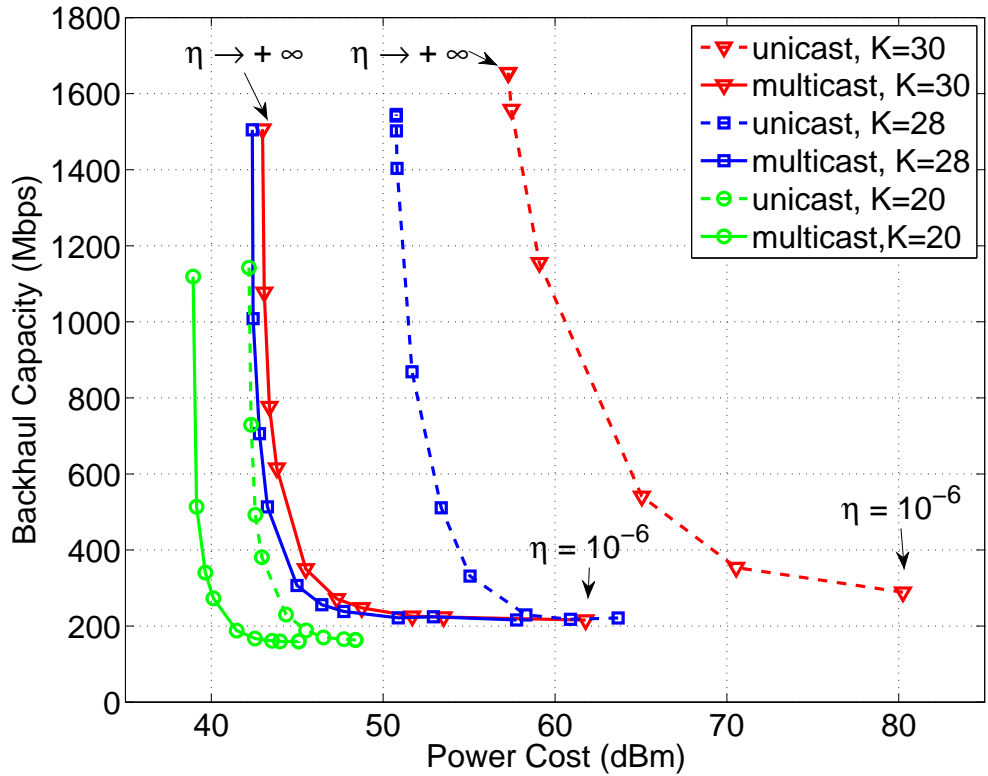}
 \caption{\small{Performance comparison between multicast transmission and unicast transmission.}}\label{fig:multicast_unicast}
\end{centering}
\end{figure}

\subsection{Effectiveness of Smooth Approximation}

 In this subsection we validate the effectiveness of the proposed smooth approximation. The ``global optimal'' solution (subject to the rank-1 condition) obtained by exhaustive search as mentioned in Section \ref{sec:P0} is considered as a benchmark. The cache-aware greedy BS clustering algorithm extended from \cite{usercentric_ICCC15} is also compared. Due to the significantly high computational complexity of the exhaustive search, we are only able to conduct the simulation in a small network with $N=3$ BS each having $L=3$ antennas, $K=6$ users, and $F=4$ files. The content popularity distributions are $\{ 0.48, 0.24, 0.16, 0.12\}$ and each BS only caches the $2$ most popular contents. The performance comparison is depicted in Fig.~\ref{fig:exhaust}. Here, the G-CCP algorithm with arctangent smooth function is simulated and each result is obtained by averaging over $200$ independent simulation trials.

\begin{figure}[h]
\begin{centering}
\includegraphics[scale=.45]{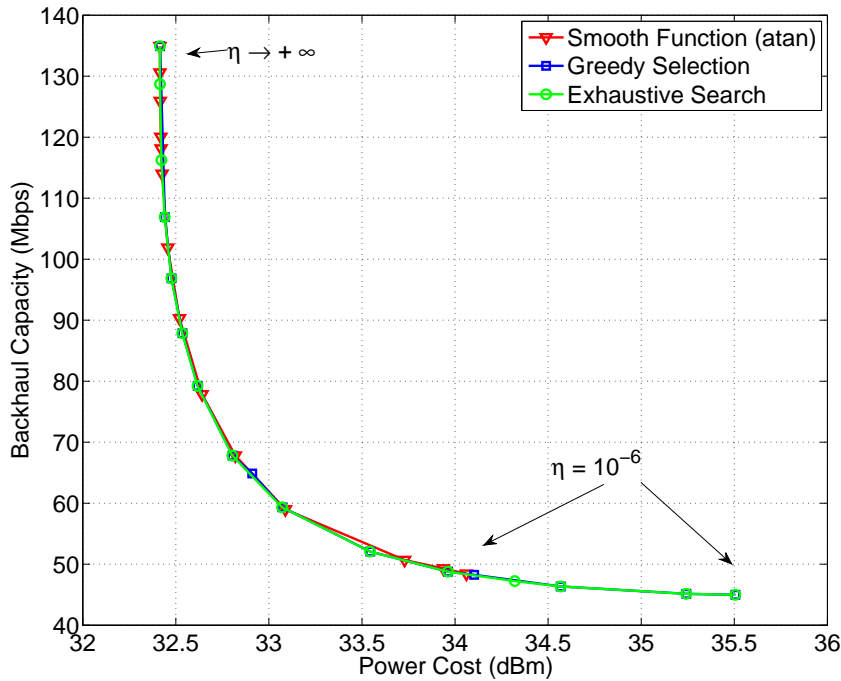}
 \caption{\small{Comparison between smooth approximation, greedy algorithm, and exhaustive search in a small network with $N=3$ BS, $K=6$ users, and $F=4$ files. }}\label{fig:exhaust}
\end{centering}
\end{figure}

From Fig.~\ref{fig:exhaust}, we observe that the performance of the smooth-function based sparse beamforming algorithm is very close to the global optimal solution, except having slightly higher minimum backhaul cost at $\eta\to 0$. This confirms the accuracy of the proposed smooth approximation at least in the considered small network. It is also observed that the greedy algorithm performs almost identically with the exhaustive search in the given scenario.

Note that when the number of base stations or users increases, the number of non-convex QCQP problems to solve in the exhaustive search and the greedy algorithm grows exponentially and quadratically, respectively, in general. However, these system parameters have no direct impact on the number of convex QCQP problems to solve in the proposed G-CCP algorithm.

\subsection{Effect of Peak Power Constraints}
Finally, we show the effect of the realistic per-antenna or per-BS peak power constraints as mentioned in Section III-A. We use the same simulation setting in Fig.~\ref{fig:alpha1} but with PopC, $K=30$ users, and one channel realization only for simplicity. The per-antenna peak power is set as 35dBm and 40dBm. The per-BS peak power is set to be $L$ (each BS has $L$ antennas) times of per-antenna peak power. It is observed from Fig.~\ref{fig:powerconstraned} that when $\eta$ is small (i.e., the backhaul has more weight), imposing peak power constraints suffers additional backhaul cost. More stringent peak power constraints result in higher minimum achievable backhaul cost. This is because the peak power constraints reduce the freedom to trade the total power for backhaul. It is also seen that when $\eta$ is large (i.e., the total transmit power has more weight), the given peak power constraints are mostly inactive and hence have little impact on the performance.

\begin{figure}[h]
\begin{centering}
\includegraphics[scale=.45]{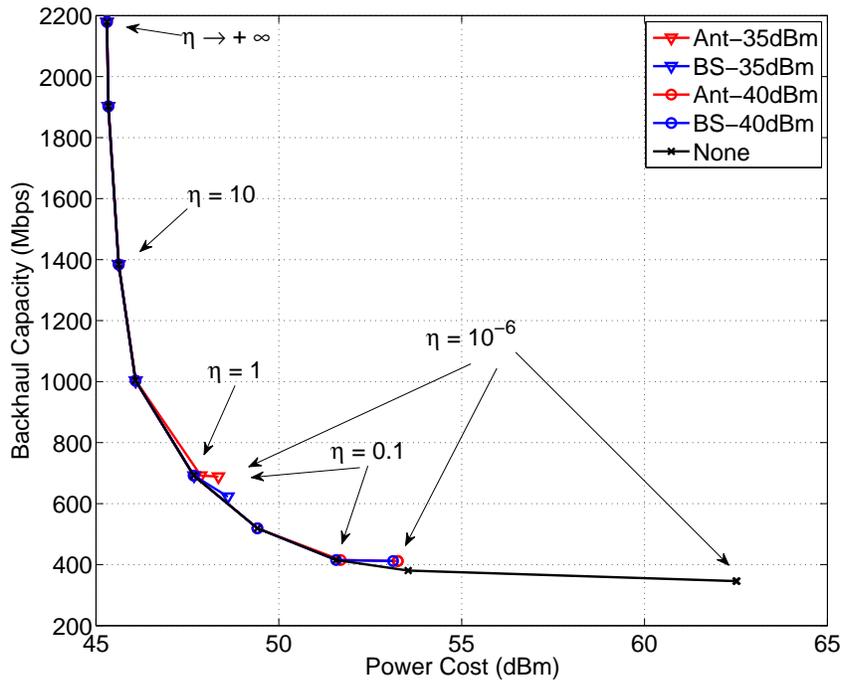}
 \caption{\small{Performance comparison under per-antenna constraints, per-BS constraints and no power constraint. }}\label{fig:powerconstraned}
\end{centering}
\end{figure}

\section{Conclusion} \label{section:Conclusion}
This paper investigates the joint design of content-centric BS clustering and multicast beamforming design in the cache-enabled cloud RAN  for wireless content delivery. We  first formulated an MINLP problem with the objective of minimizing the total network cost, modeled by the weighted sum of backhaul cost and transmit power cost, subject to the QoS constraint for each multicast group. %
Based on such problem formulation we have proved that all the BSs which cache the content requested by a multicast group can be always included in the BS cluster of this content, regardless of their channel conditions.
To make the problem more tractable, we reformulated the equivalent sparse beamforming design problem which takes the dynamic BS clustering inexplicitly. By adopting the smoothed $\ell_0$-norm approximation, we then converted the sparse beamforming problem into two forms of DC programs.
The first form has DC objective and convex constraints by using SDR approach and SDR-based CCP algorithm is introduced to find the local optimal solution.
The second DC program is in the general form with DC expressions in both objective and constraints, and is solved using the generalized CCP algorithm. Simulation results showed that the generalized-CCP based sparse beamforming algorithm outperforms the SDR-CCP based algorithm in both transmission power efficiency and computation efficiency.  We also compared three heuristic caching strategies by simulation in the cache-enabled cloud RAN, random caching, popularity-aware caching, and probabilistic caching. It is shown that popularity-aware caching in general provides the best tradeoff curves between backhaul cost and transmit power cost. But in the extreme case where only backhaul cost is considered in the total network cost, the probabilistic caching outperforms the popularity-aware caching when user density is low. Last but not least, simulation results demonstrated that the proposed content-centric transmission (i.e. content-centric BS clustering and multicast beamforming) offers significant reduction in total network cost than the conventional user-centric design (i.e. user-centric BS clustering and unicast beamforming) under the considered content-request model.

This work can be viewed as an initial attempt from the physical layer toward the design of content-centric wireless networks. There are many interesting directions to pursue based on this work. For instance, the problem formulation in this work assumes that the content placement is given and fixed. As observed by simulation, the performance of the considered heuristic caching strategies depends on not only the content popularity,
but also the user density as well as the network optimization objective. As such, it is of particular importance to investigate the globally optimal caching strategy through the joint design of mixed timescale cache placement/replacement and content delivery. In addition, the proposed sparse multicast beamforming algorithms are centralized and may be difficult to implement in very large networks. To be more practical in cloud RAN architectures with large number of users, low-complexity or distributed implementation of these algorithms is greatly desired.

\section*{Appendix A: Proof of Proposition  \ref{Proposition-P1} }
Assume that the content requested by user group $m^\ast$ is already cached in BS $n^\ast$, i.e., $c_{f_{m^\ast},n^\ast}=1$.  Consider an arbitrary BS clustering matrix $\mathbf{S}'$ with $ s'_{m^\ast,n^\ast} = 0$. The minimum total network cost incurred by the given $\mathbf{S}'$ is denoted as $C'_N = C'_B + \eta C'_P$, where $C_B'$ is determined by \eqref{eqn:backhaul cost} with the summation term $s_{m^\ast,n^\ast}(1-c_{f_{m^\ast}, n^\ast})R_{m^\ast}$ being zero, and $C'_P$ is the optimal solution of the following power minimization problem at given $\mathbf{S}'$
\begin{subequations}
\begin{align}
\mathcal{P}(\mathcal{Z}_{\mathbf{S}'}): ~C'_P=\mathop{\text{minimize}}_{ \{ \mathbf{w}_{m,n}\} } \quad & \sum_{m=1}^M \sum_{n=1}^N \lVert \mathbf{w}_{m,n} \rVert_2^2 \\
\text{subject to} \quad  &\eqref{eqn:SINR-constraint} \nonumber \\
& \mathbf{w}_{m,n} = 0,~\forall (m,n) \in \mathcal{Z}_{\mathbf{S}'} \label{eqn:P1-cluster}
\end{align} \label{eqn:power-minimization-P1}
\end{subequations}
Here, $\mathcal{Z}_{\mathbf{S}'}$ is the set of inactive BS-content associations at given $\mathbf{S}'$, i.e., $\mathcal{Z}_{\mathbf{S}'} = \{ (m,n)| [\mathbf{S}']_{m,n}=0\}$. Obviously we have $(m^\ast,n^\ast) \in \mathcal{Z}_{\mathbf{S}'}$.

Now, define a new BS clustering matrix $\mathbf{S}''$ which only differs from $\mathbf{S}'$ at the $(m^\ast,n^\ast)$-th element, i.e., $s''_{m^\ast,n^\ast}=1$. Then the minimum total network cost incurred by  $\mathbf{S}''$ can be written as $C''_N = C''_B + \eta C''_P$. Here, $C''_B$ is also determined by \eqref{eqn:backhaul cost} with the summation term $s_{m^\ast,n^\ast}(1-c_{f_{m^\ast}, n^\ast})R_{m^\ast}$ being zero again due to $c_{f_{m^\ast}, n^\ast}=1$ and hence one has $C''_B=C'_B$. $C''_P$ is the optimal solution of the following power minimization problem at given $\mathbf{S}''$
\begin{subequations}
\begin{align}
\mathcal{P}(\mathcal{Z}_{\mathbf{S}''}):~C''_P=\mathop{\text{minimize}}_{ \{ \mathbf{w}_{m,n}\} } \quad & \sum_{m=1}^M \sum_{n=1}^N \lVert \mathbf{w}_{m,n} \rVert_2^2 \\
\text{subject to} \quad  &\eqref{eqn:SINR-constraint}  \nonumber \\
& \mathbf{w}_{m,n} = 0,~\forall (m,n) \in \mathcal{Z}_{\mathbf{S}''} \label{eqn:P2-cluster}
\end{align} \label{eqn:power-minimization-P2}
\end{subequations}
By definition, the set of inactive BS-content association at $\mathbf{S}''$ satisfies:
\begin{equation}
\mathcal{Z}_{\mathbf{S}''} = \mathcal{Z}_{\mathbf{S}'} \backslash (m^\ast,n^\ast). \label{eqn:cluster_set}
\end{equation}
By observing $\mathcal{P}(\mathcal{Z}_{\mathbf{S}'})$ and $\mathcal{P}(\mathcal{Z}_{\mathbf{S}''})$ closely, it is clear that the feasible set of $\mathcal{P}(\mathcal{Z}_{\mathbf{S}'})$ is only a subset of that of $\mathcal{P}(\mathcal{Z}_{\mathbf{S}''})$ due to \eqref{eqn:cluster_set}. Therefore, we have $C'_P \geq C''_P$. Together with $C'_B = C''_B$, we obtain that the two total network costs satisfy $C'_N \geq C''_N$.
 This means that, for any BS clustering matrix $\mathbf{S}'$ with $s'_{m^\ast,n^\ast} = 0$, we can always find another BS clustering matrix $\mathbf{S}''$ with $s''_{m^\ast,n^\ast} = 1$ such that it can achieve a total network cost no larger than that of $\mathbf{S}'$. Therefore, without loss of optimality, one can set $s_{m^\ast,n^\ast} = 1$ in $\mathcal{P}_0$.

\bibliographystyle{IEEEtran}
\bibliography{IEEEabrv,content-centric-beamforming}

\end{document}